\pdfoutput=1
\documentclass[twocolumn, floatfix, nofootinbib]{revtex4-2}
\usepackage[utf8]{inputenc}
\usepackage[english]{babel}
\usepackage{graphicx}
\usepackage{placeins}
\usepackage{amsmath, amssymb, amsthm}
\usepackage{tensor}
\usepackage{braket}
\usepackage{cases}
\usepackage{dsfont}
\usepackage{hyperref}
\usepackage[percent]{overpic}
\usepackage{svg}
\usepackage[export]{adjustbox}
\usepackage{leftidx}
\usepackage{BOONDOX-cal} 
\usepackage{bbm} 
\usepackage{orcidlink}

\numberwithin{equation}{section}

\DeclareMathOperator{\tr}{tr}
\DeclareMathOperator{\diag}{diag}
\DeclareMathOperator{\hspan}{span}
\DeclareMathOperator{\one}{\mathbbm{1}}

\newtheorem{defi}{Definition}[section]

\newtheorem{prop}[defi]{Proposition}
\newtheorem{coro}[defi]{Corollary}
\newtheorem{exa}[defi]{Example}
\newcommand{\comm}[2]{\left[ #1 , #2 \right]}
\newcommand{\acomm}[2]{\left[ #1 , #2 \right]_+}

\newcommand{\cre}[2]{c^\dagger{}^{#2}(#1)}
\newcommand{\ann}[2]{c_{#2}(#1)}

\newcommand{\vac}{\left\lvert 0 \right\rangle}

\newcommand{\SprojE}[0]{S_{\mathcal{S}}}
\newcommand{\SprojF}[0]{2 S_{E'} \cdot \Sx{e_{n+1}(0)}{}}
\newcommand{\Sx}[2]{S^{#2}({#1})}
\newcommand{\Sxl}[2]{S_{#2}({#1})}
\newcommand{\Se}[2]{S_{#1}^{#2}}
\newcommand{\SE}[2]{S_{#1}^{#2}}

\newcommand{\jrep}[2]{\pi_{#1}({#2})}

\newcommand{\scpr}[2]{\langle#1\, \vert \, #2 \rangle}

\newcommand\der[1]{\ensuremath{\text{d} #1}}
\newcommand{\cm}[0]{\ensuremath{\text{cm}}}

\begin{document}
\title{Fermion Spins in Loop Quantum Gravity}

\author{Refik Mansuroglu \orcidlink{0000-0001-7352-513X}}
\email[]{Refik.Mansuroglu@fau.de}
\affiliation{Institute for Quantum Gravity,
Friedrich-Alexander-Universität Erlangen-Nürnberg (FAU), Staudtstraße 7, 91058 Erlangen, Germany}
\author{Hanno Sahlmann \orcidlink{0000-0002-8083-7139}}
\email[]{Hanno.Sahlmann@gravity.fau.de}
\affiliation{Institute for Quantum Gravity,
Friedrich-Alexander-Universität Erlangen-Nürnberg (FAU), Staudtstraße 7, 91058 Erlangen, Germany}

\date{\today}

\begin{abstract}
 We define and study kinematical observables involving fermion spin, such as the total spin of a collection of particles, in loop quantum gravity.
 Due to the requirement of gauge invariance, the relevant quantum states contain strong entanglement between gravity and fermionic degrees of freedom. Interestingly we find that properties and spectra of the observables are nevertheless similar to their counterparts from quantum mechanics. 
 However, there are also new effects. Due to the entanglement between gravity and fermionic degrees of freedom, alignment of quantum spins has consequences for quantized geometry. We sketch a particular effect of this kind that may in principle be observable.  
\end{abstract}

\maketitle

\section{Introduction}
Spin is a quintessential quantum property but at the same time a bit of geometric information -- a direction. Indeed, spins were at the heart of Penrose's spin networks \cite{Penrose71angularmomentum:,Penrose72:}, invented as hypothetical quantum states of spatial geometry. This was indeed prophetic, as in loop quantum gravity (for reviews see \cite{Ashtekar:2004eh,Thiemann:2007pyv}), spin networks turn out to describe a basis of eigenstates of spatial geometry \cite{Baez:1995md}. With the present work, we will, in a sense, close a circle, as we will investigate the quantum theory of particle spins in loop quantum gravity. It will turn out that one can regard spin networks as flows of physical spin with some justification. 

The quantum theory of spinor fields coupled to loop quantum gravity is by now fairly well understood \cite{MoralesTecotl:1994ns,Morales1995,Baez:1997bw,Thiemann:1997rq,Thiemann:1997rt,Mercuri:2006um,Bojowald:2007nu,Bojowald:2008nz}.
In particular, we will follow Thiemann \cite{Thiemann:1997rq} in adopting a discretisation of the classical symplectic structure that is extremely natural in the context of loop quantum gravity. As a consequence, the space of fermion states can be described as a fermionic Fock space of particle-like excitations. 

Loop quantum gravity uses a formulation of general relativity as a constrained gauge theory, with a Gau{\ss} constraint encoding SU(2) gauge invariance (invariance under spatial frame rotations). It was realised early on \cite{MoralesTecotl:1994ns, Morales1995, Baez:1997bw} that, to solve the Gau{\ss} constraint, gravity and fermionic excitations have to be coupled. A very compelling solution was first suggested in \cite{MoralesTecotl:1994ns,Morales1995} and later expanded on in \cite{Baez:1997bw}: the fermions sit at the open ends of gravitational spin networks. More recently, the coupling of fermion states to gravity has been investigated also from the perspective of spin foam models for loop quantum gravity \cite{Han_2013, Bianchi_2013}. The coupled gravity-fermion states we are considering here are precisely the boundary states in the spinfoam formalism.

In the present work, we set out to define and investigate observables that measure the total fermion spin contained in spatial regions. Depending on the quantum state, this can be the spin of a single fermion, or the total spin of multi-fermion configurations. Specifically, we define operators measuring the squared total spin, and components of the total spin in the direction of another spin, or the component normal to a surface. 

To add spins, we need to transport them through space. We use the Ashtekar-Barbero connection for this. Consequently the spin operators modify fermion and gravitational states in tandem. Somewhat surprisingly, algebra and spectra have much in common with the situation in quantum mechanics. However, there are also new effects. For example we show that alignment of spins in a region can lead to a change of its surface area, which could in principle be observable. It also shows that one can regard spin networks as flows of physical spin with some justification. 

When we speak of spin observables, we have to be careful, however. 
In the present work we will not be concerned with implementing diffeomorphism and Hamilton constraint which encode general covariance. In particular the spin observable that we consider are not observables in the sense of Dirac, and it is in fact not clear what ultimate significance they have. In this way, they are similar to the geometric operators for area, volume and length. It should be fairly straightforward to generalise some of the results of the present work to the spatially diffeomorphism invariant level. Moreover, if one considers material reference systems as described in \cite{Giesel:2007wn,Giesel:2012rb}, we expect that the operators we describe correspond to physical quantities.

An observation that is crucial for the construction of the spin observables is that the spatial frame provides an isomorphism between tangent spaces and the internal space. This means that in principle we can freely move between spins as tangent vectors and spins as internal vectors, and scalar products can be evaluated in tangent space using the metric or internally. In practice, the latter picture is the only practical one in loop quantum gravity, as no operator for the spatial metric is known, and neither are any other operators for tensorial quantities.
\newpage

Throughout the paper, we will follow a number of conventions. Spacetime is described by a smooth manifold $M$ foliated into 3-dimensional hypersurfaces $M=\mathbb{R}\times \Sigma$. For the metric signature we use the ``mostly plus'' convention $\eta = \diag(-,+,+,+)$. 
We use capital latin letters $I,J,\ldots = 0, 1, 2, 3$ for indices of tensor fields in the 4-dimensional internal Minkowski space and lower case latin letters $i,j, \ldots = 1, 2, 3$ for indices of tensor fields in the spatial internal space, which are pulled with the Euclidean metric $\delta$. Tensor fields on $M$ carry indices described by lower case greek letters $\mu,\nu,\ldots = 0, 1, 2, 3$, tensor fields on $\Sigma$ carry indices described by lower case latin letters $a,b,\ldots = 1, 2, 3$. Finally, Weyl spinors and spin $\frac{1}{2}$ representations of SU$(2)$ carry indices described by capital latin letters $A,B,\ldots = 1, 2$. We define holonomies $h_e$ such that they transform under gauge transformations as
\begin{align}
    h_e \mapsto g(s(e)) \, h_e \, g^{-1}(t(e))
\end{align}
where $t(e)$ is the endpoint and $s(e)$ the starting point of the edge $e$.

\section{A Quick Review on Loop Quantum Gravity with Fermions}
On classical level, a theory of spin $\frac{1}{2}$ fermions coupled to gravity is described by the Einsten-Cartan-Holst action \cite{Thiemann:1997rt,Mercuri:2006um,Bojowald:2007nu} together with the covariant version of the Dirac action 
\begin{align}
    S[e,\omega,\Psi] &= \frac{1}{16 \pi G} \int_M \text{d}x^4\; |\det{e}| e^{\mu}_Ie^\nu_J P^{IJ}{}_{KL} F^{IJ}_{\mu\nu}(\omega) \nonumber \\
    &+\frac{i}{2} \int_M \text{d}x^4\; |\det{e}|\left[\overline{\Psi} \gamma^Ie_I^\mu \nabla_\mu \Psi - c.c.\right].
    \label{eq:action}
\end{align}
A foliation into three-dimensional hypersurfaces $M=\mathbb{R}\times \Sigma$ is performed, and the canonical analysis yields the canonical gravitational variables \cite{Thiemann:1997rt, Bojowald:2007nu},
\begin{align}
    \mathcal{A}_a^i(x) &=\Gamma_a^i +\beta K_k^i + 2\pi G\beta \epsilon^i{}_{kl} e^k_a J^l \label{eq:connection} \\ 
    E^a_i(x) &=\frac{1}{2}\epsilon_{ijk}\epsilon^{abc} e^i_ae^j_b(x),
\end{align}
which is the analogue of the Ashtekar connection and its conjugate momentum. In (\ref{eq:connection}), $\Gamma$ is the torsion free spin connection, $K$ the torsion free extrinsic curvature and $J^l$ are the spatial components of the fermion current,
\begin{equation}
    J^l=\overline{\Psi}\gamma^l\Psi.
\end{equation}
Note that, in contrast to matter-free loop quantum gravity, the Ashtekar connection carries non-zero torsion. In the chiral basis, we can split the Dirac fermion $\Psi$ and its conjugate momentum into its chiral components and define the half-densities \cite{Thiemann:1997rt},
\begin{align}
    \theta^A(x) &=\sqrt[4]{\det q} \;\Psi_\text{R}(x), \qquad \pi_\theta(x)=-i\theta^\dagger(x)\\
    \nu_A(x)&=\sqrt[4]{\det q}\; \Psi_\text{L}(x), \qquad \pi_\nu(x)=-i\nu^\dagger(x)
\end{align}
The non-vanishing anti-Poisson relations for the Weyl spinors are then induced by the anti-Poisson relations of $\Psi$. They read  
\begin{equation}
\label{eq:antipoi}
    \left\{\theta^A(x), \pi_{\theta B}(y)\right\}_+= \delta^{A}_{B} \, \delta^{(3)}(x,y),
\end{equation}
and similar for $\nu$. In the following, we will focus on one Weyl component and its momentum $(\theta, \pi_\theta)$ only. However, everything works analogously for the other chiral component.

The action (\ref{eq:action}) yields contributions to Gauss, diffeomorphism and Hamilton constraint \cite{Thiemann:1997rt, Bojowald:2007nu}. For our purposes, only the Gauss constraint 
\begin{equation}
\label{eq:gauss}
    G_i(x)=D^{(\mathcal{A})}_aE^a_i(x) +
    \left( \theta^\dagger(x)\sigma_i\theta(x) + \nu^\dagger(x)\sigma_i\nu(x)\right)
\end{equation}
will be relevant. The first term of (\ref{eq:gauss}) is the SU(2) Gauss constraint known from matter-free loop quantum gravity. The second term is the SU(2) current of the half-density Weyl spinors, which generates SU(2) transformations.

The Hilbert space of gravitational degrees of freedom is constructed the same way as in the matter-free theory. The gravitational observables hence act on cylindrical functions which form the Ashtekar-Lewandowski Hilbert space,
\begin{equation}
    \mathcal{H}_\text{AL}=L^2(\Sigma,\text{d}\mu_{\text{AL}}),
\end{equation}
via multiplication of holonomies and the interior product with the derivation $X_S$, respectively \cite{Ashtekar:1996eg}
\begin{align}
\label{eq:Grep}
    \jrep{j}{h}_e \Psi[A] &= \jrep{j}{h}_e[A] \Psi[A], \\
    \int_S E_i \Psi[A] &= i (X_S^i\Psi)[A].
\end{align}
The Weyl fields on the other hand are quantized as fermionic creation and annihilation operators,
\begin{equation}
\label{eq:Frep}
    \theta^A(x) = \cre{x}{A},  \qquad -i \pi_{\theta B}(y)= \ann{y}{A}
\end{equation}
which satisfy canonical anticommutation relations,
\begin{align}
    \acomm{\ann{x}{A}}{\cre{y}{B}} &= \delta_{x,y}\,\delta^A_B, \\
    \acomm{\cre{x}{A}}{\cre{y}{B}}=0, \qquad &\acomm{\ann{x}{A}}{\ann{y}{B}} = 0.
\end{align}
The action of this operator spans the fermionic Fock space over a one particle space given by 
\begin{align}
    \mathcal{h} &= \{f:\Sigma \longrightarrow\mathbb{C}^2:\; f(x)\neq 0 \text{ only for finitely many } x \} \nonumber \\
    &\scpr{f}{f'}=\sum_{x\in\Sigma}\overline{f(x)} f'(x)
\end{align}
We note the change between \eqref{eq:antipoi} and 
\begin{equation}
    \acomm{\theta(x)^A}{\pi_{\theta B}(y)}= i\delta_B^A\delta_{x,y},
\end{equation}
in that a Dirac delta has become a Kronecker delta. This means that, following \cite{Thiemann:1997rq}, and in contrast to \cite{Baez:1995md} we quantise a modified symplectic structure.

The combined system of matter fields and gravitational degrees of freedom is given by the tensor product of the Ashtekar-Lewandowski Hilbert space and the fermionic Fock space over $\mathcal{h}$
\begin{equation}
    \mathcal{H}=\mathcal{H}_\text{AL}\otimes \mathcal{F}_-(\mathcal{h}),
\end{equation}
so we extend \eqref{eq:Grep}, \eqref{eq:Frep} in the obvious way.

On the way towards a physical Hilbert space, we need to implement the Gauss constraint. For this, we only regard quantum states lying in the kernel of the quantum Gauss constraint operator, or equivalently states that are invariant with respect to the unitary action $U_g$ of gauge transformations generated by the Gauss constraint operator. The action of $U_g$ reads\footnote{Note that we are using the convention that the first index of $h_e$ transforms at $s(e)$ and the second one at $t(e)$. At the same time, we are using Thiemann's convention for $e \circ f$ \cite{Thiemann:2007pyv} such that $h_e \cdot h_f = h_{e \circ f}$ holds.}
\begin{align}
    U_g\jrep{j}{h_e} U^{-1}_g &= g(s(e)) \cdot\jrep{j}{h_e} \cdot g^{-1}(t(e)),\\
    U_g \theta(x) U^{-1}_g &= g(x) \cdot \theta(x), \\
    U_g \pi_{\theta} (x) U^{-1}_g &= \pi_{\theta(x)}\cdot g^{-1}(x).
\end{align}
We finally arrive at the Hilbert space of SU(2) invariant states. One suitable basis is given by a generalisation of spin network states, which also admit Weyl spinors at the vertices of the underlying spin network graph.

\section{Fermion Spin Observables}
We want to define observables based on the spin of fermions. To this end, we first have to clarify the definition of angular momentum and spin for fermion fields in curved spacetime. In flat space, angular momentum is described by the covariant angular momentum bivector $J^{\mu\nu}$. For a point particle with position $x^\mu$ and momentum $p_\mu$ this is given by 
\begin{equation}
    J^{\mu\nu}=x^\mu p^\nu-p^\mu x^\nu,
\end{equation}
where here and in the following few equations we are pulling indices with the Minkowski metric.
In quantum theory, the corresponding operators are nothing but the generators of Lorentz transformations. In quantum field theory, these generators are given as spacetime integrals of angular momentum densities, and for non-scalar fields, a contribution from the intrinsic spin appears.   
For a Dirac field $\Psi(t,x)$, 
\begin{align}
\label{eq:ang_mom_flat}
  J^{\mu\nu} = &\int\text{d}^3x\;\overline\Psi(t,x)\,(-i)\gamma^0 \times \nonumber \\ 
  &\times \left(\comm{x^\mu}{p^\nu}+ \frac{i}{4} \comm{\gamma^\mu}{\gamma^\nu}\right)\Psi(t,x). 
\end{align}
The two terms correspond to orbital angular momentum and spin, respectively. Fixing a reference system by specifying a timelike normal $n$, one can recover spatial angular momentum as
\begin{equation}
    J^\mu= (*J)^{\mu\nu}n_\nu\equiv \frac{1}{2} \epsilon_{\mu'\nu'}{}^{\mu\nu} J^{\mu'\nu'} n_\nu. 
\end{equation}
This is a spatial vector,  $J^\mu n_\mu=0$, so in adapted coordinates $J^a, a=1,2,3$ is the angular momentum vector. In particular, for the spin component,
\begin{align}
  S^{\mu\nu} &= \int\text{d}^3x\;\frac{1}{4}\overline\Psi(t,x)\,\gamma^0 \comm{\gamma^\mu}{\gamma^\nu}\Psi(t,x),\\ 
  S^a &= \int\text{d}^3x\;\frac{1}{8}\overline\Psi(t,x)\,\gamma^0 \epsilon^a{}_{bc}\comm{\gamma^a}{\gamma^b}\Psi(t,x),\\
  S^{\mu\nu} \vac &= 0.
\end{align}
In the chiral representation, the spin operator works out to 
\begin{equation}
\label{eq:spin_flat}
    S^a= \int\text{d}^3x\;\overline\Psi(t,x)\,\frac{1}{2}\begin{pmatrix}\sigma^a&0\\0&-\sigma^a \end{pmatrix}\,\Psi(t,x), 
\end{equation}
which is nothing but the generator, on spinor space, of rotations as defined by $n=(1,0,0,0)$.

In curved spacetime there are generically no isometries, and so no action of global Lorentz transformations. Thus there is no obvious definition of \emph{total} angular momentum or spin. However, there is a close analogue to spin \emph{densities} in \eqref{eq:ang_mom_flat}, \eqref{eq:spin_flat}: There are generators for \emph{local} Lorentz transformations (i.e. local frame rotations), and those can provide a spin bivector density.

To obtain Ashtekar variables, a partial gauge fixing is employed. A timelike unit vector $n^I(x)$ is introduced to fix three components of the frame $e^\mu_I$. 
$n^I(x)$ introduces a decomposition of the local Lorentz transformations into rotations and boosts, and only local rotations survive as gauge transformations. Their generator is the Gauss constraint \eqref{eq:gauss}, and so a natural candidate for the local spin density contribution from the Weyl field $\theta$ is given by 
\begin{equation}
    \Sx{x}{i}= \frac{1}{2} \theta(x) \sigma_i \theta^\dagger(x).
\end{equation}
It immediately follows that
\begin{align}
        \comm{\Sx{x}{i}}{\cre{y}{A}}= \frac{1}{2}\delta_{x,y}\, \sigma^i{}^A{}_{B}\,\cre{y}{B}, \\
        \Sx{x}{i} \vac =0, 
\end{align}
and 
\begin{equation}
    \comm{\Sx{x}{i}}{\Sx{y}{j}}=i\delta_{x,y}\, \epsilon^{ij}{}_k\,  \Sx{x}{k}.
\end{equation}
Note that while $\Sx{x}{i}$ is an internal vector field, we can change freely between internal and tangent space using the 3d frame $e$, 
\begin{equation}
    \Sx{x}{a}(x):= e^a_i(x)\Sx{x}{i}
\end{equation}
and scalar products can be taken in either space, 
\begin{equation}
    \Sx{x}{a}\Sx{x}{b} q_{ab}(x)= \Sx{x}{i}\Sx{x}{j}\delta_{ij}.  
\end{equation}
For the quantum theory, working in the internal space has advantages, as there are no states known that transform like tensors in tangent space, whereas states that transform nontrivially under gauge transformations are well studied. 

We will also need to transport spins from one point to the other to compare them. This raises the question how to translate between transport in tangent and internal space, and which connection to use for the transport. 
The first question is easy to answer. If $B$ is a connection in tangent space and $C$ is a connection in internal space, then the condition 
\begin{equation}
    D_a e_b^i\equiv \partial_a e_b^i -B_a{}^c{}_b\; e_c^i + C_a{}^i{}_j \;e_b^j=0
\end{equation}
is equivalent to 
\begin{equation}
    h_e^{(B)}{}^a{}_b\;e^b_i(t(e))=e_j^b(s(e))\;h_e^{(C)}{}^j{}_i. 
\end{equation}
Thus, fixing a connection on one bundle will yield a compatible connection on the other. The second question is more subtle. The spin connection of $e$ (or equivalently the 3d Levi-Civita connection $\Gamma$) or the Ashtekar-Barbero connection $A$ immediately come to mind, but there are undoubtedly others. Ultimately, the connection used is part of the choice of observables that we will define below. We will make use of the the Ashtekar-Barbero connection $A$, as it is a well defined operator in loop quantum gravity. Then we can define the parallel transport of spin operator as 
\begin{equation}
    \Se{e}{i}= \jrep{1}{h_e^{-1}}{}^{i}{}_j\,\Sx{s(e)}{j} 
\end{equation}
where $s(e)$ is the initial point (source) of $e$. Interestingly, the parallel transported spin is again a quantum theoretical spin:
\begin{equation}
    \comm{\Se{e}{i}}{\Se{e}{j}}=i \epsilon^{ij}{}_k\,  \Se{e}{k}.
\end{equation}
Under a gauge transformation $g$ it transforms as 
\begin{equation}
    U_g\, \Se{e}{i}\, U_{g}^{-1}= \jrep{1}{g}(t(e))^i{}_j  \Se{e}{j}
\end{equation}
where $t(e)$ is the final point (target) of $e$. 

\subsection{Total Spin}

\noindent In analogy to the theory of angular momentum in flat standard quantum theory, we start by defining the total spin of the fermion field evaluated at a number of points. To add the spins, they all have to be transported to the same point. To do so, one has to chose a set of edges $E$ with common endpoint and disjoint initial points. 
\begin{defi}
Given a set of edges $E$ with common endpoint and disjoint initial points, $t(e)=t(e')$ and  $s(e)\neq t(e')$ for all $e,e'\in E$, define the total spin 
\begin{equation}
    \SE{E}{}=\sum_{e\in E}   \Se{e}{}
\end{equation}
\end{defi}
For a graphical representation of the total spin see figure \ref{fig:total_spin}. 
\begin{figure}
    \centering
    \includegraphics[scale=0.1]{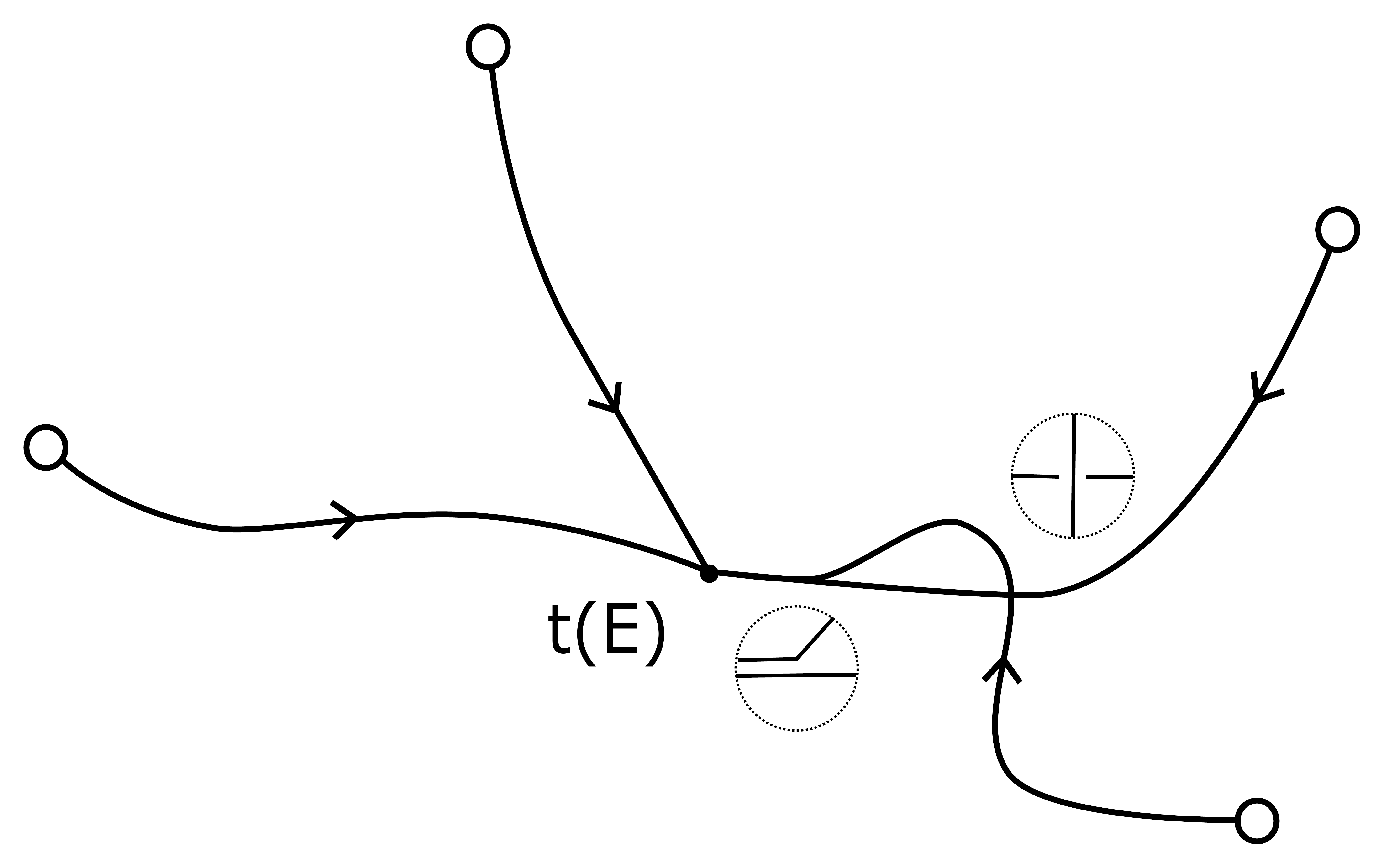}
    \caption{Graphical representation of the total spin operator $\Se{E}{}$. Note that edges can intersect and also run on top of each other, as long as the corresponding intertwiners are trivial. This is indicated by the virtual networks drawn next to the intersections in the figure.}
    \label{fig:total_spin}
\end{figure}
The total spin $\SE{E}{}$ is again a quantum mechanical spin, 
\begin{equation}
    \comm{\SE{E}{i}}{\SE{E}{j}}=i \epsilon^{ij}{}_k\,  \Se{E}{k}.
\end{equation}
Under a gauge transformation $g$ it transforms at the joint final point $t(E)$ of the $e\in E$, 
\begin{equation}
    U_g\, \Se{E}{}\, U_{g}^{-1}= \jrep{1}{g}(t(E)) \cdot \SE{E}{j}. 
\end{equation}
The total spin is not gauge invariant, but it can be used to construct a large class of gauge invariant observables. 

Invariant subspaces for the action of $\Se{E}{}$ are constructed as follows. Denote by $E'\subseteq E$ any subset of $E$ and the edges in $E'$ as  $E'=\{e_1,e_2, \ldots e_{|E'|}\}$. Define states,
\begin{equation}
    \Psi_{E'}{}^{A_1}\cdots ^{A_{|E'|}}=\prod_{k=1}^{|E'|} \jrep{\frac{1}{2}}{h_{e_k}}^{A_k}{}_{B_k} \otimes  \theta(s(e_k))^{B_k} \vac.  
\end{equation}
They span the Hilbert space,
\begin{equation}
    \mathcal{H}_E:= \hspan \left\{ \Psi_{E'}^{A_1 \dots A_{|E'|}} \rvert A_1, \dots, A_{|E'|}\in 
    \{1,2\},\; E'\subseteq E \right\}.
\end{equation}
The space $\mathcal{H}_E$ is invariant under the action of gauge transformations. Since its states transform only at $t(E)$, the irreducible subspaces of $\mathcal{H}_E$ are effectively labeled by irreps of SU(2), 
\begin{align}
    \mathcal{H}_E= \bigoplus_j  \mathcal{H}_E^{(j)},& \qquad \mathcal{H}_E^{(j)} = \bigoplus_k^{n_j} \mathcal{H}_E^{(j)(k)} \nonumber \\
    &\mathcal{H}_E^{(j)(k)} \simeq
    \pi_j
\end{align}
with $n_j$ denoting possible multiplicity (including $n_j=0$). Note that we can also easily add two total spin operators and arrive at
\begin{align}
    \SE{E \cup E'}{i} = \SE{E}{i} + \SE{E'}{i}
\end{align}
which is again a spin operator, if and only if the set of starting points of $E$ and $E'$ are disjoint. Also, we have to impose $t(E) = t(E')$ in order to get a gauge covariant object.

Let us consider the action of $(\SE{E}{})^2$ on $\mathcal{H}_E$. We have 
\begin{align}
    (\SE{E}{})^2 &= \sum_{e\in E} \Sx{s(e)}{2} \, + \nonumber \\
    &+ \sum_{(e,e')\in E\times E } \Sx{s(e)}{}\cdot h_{e \circ e'^{-1}}\cdot \Sx{s(e')}{}.
     \label{def_Jptot2}
\end{align}
In order to understand the mixed terms, we make use of the binor calculus, a graphical notation suggested by Penrose 
\cite{Penrose71angularmomentum:, Penrose:1987uia, DePietri:1996pj, Kauffman02}. 
Using this toolbox we can rewrite the mixed terms in the following way
\begin{align}
    &2 \tensor{\Sxl{s(e_1)}{k}}{^A_B} \tensor{\left(h_{e} \right)}{^j_k} \tensor{\Sx{s(e_2)}{k}}{^C_D} = \nonumber \\[2 ex]
    &= \leftidx{^A_B}{\includegraphics[valign=c, scale=0.2]{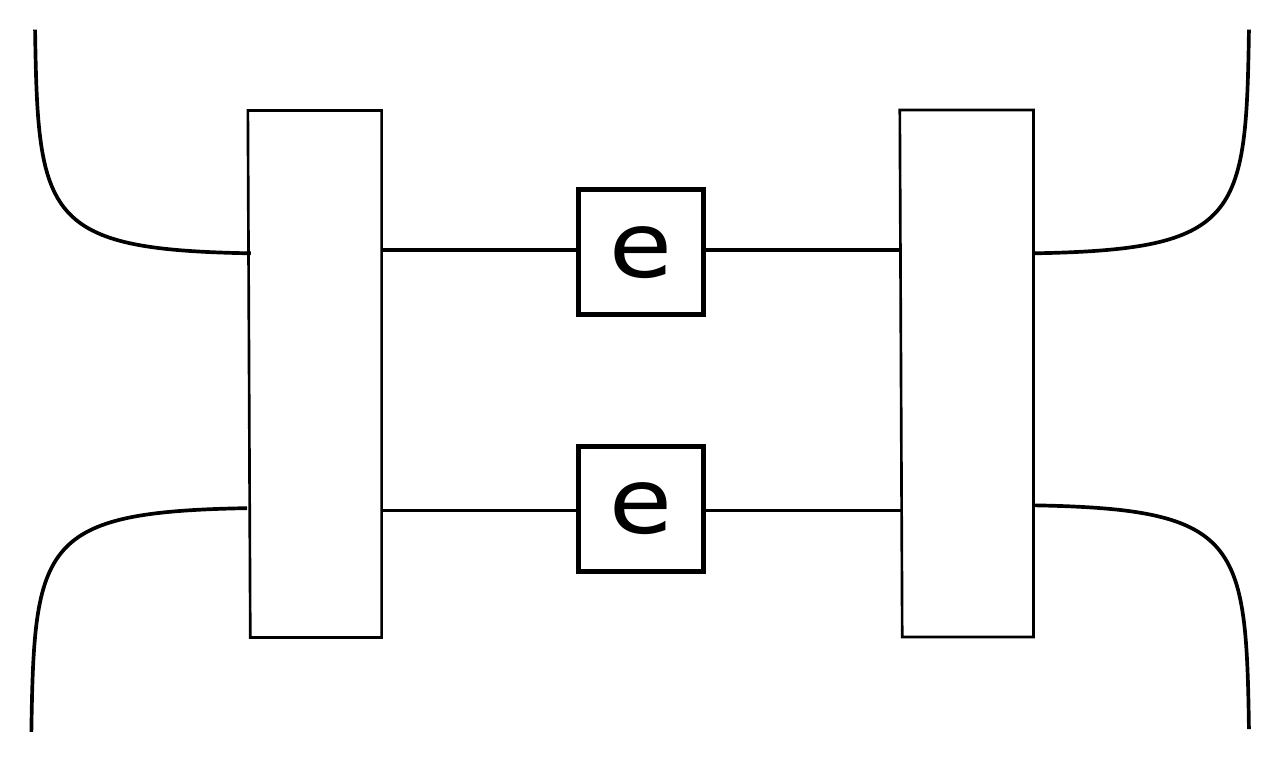}}{^C_D} \nonumber \\[3 ex]
    &= \frac{1}{2} \leftidx{^A_B}{\includegraphics[valign=c, scale=0.5]{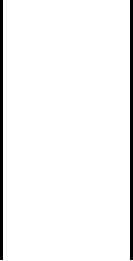}}{^C_D} + \leftidx{^A_B}{\includegraphics[valign=c, scale=0.5]{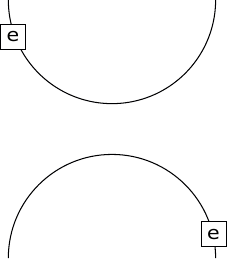}}{^C_D}, \label{PC_2J1J2}
\end{align}
where we denoted $e := e_1^{-1} \circ e_2$ for the sake of clarity. For technical details and a proof of this identity, see Appendix \ref{Penrose_Calculus_appendix}.  Now, given an arbitrary spin network state $\psi$ with $n$ fermions sitting on distinct vertices of the spin network, we can calculate the action of (\ref{def_Jptot2}) by writing $\psi$ using the binor formalism and applying (\ref{PC_2J1J2}). In general, however, the spin network states are not eigenstates of (\ref{def_Jptot2}). Instead, we find a subset of its eigenstates with eigenvalues, which are very reminiscent of flat quantum theory,
\begin{align}
   \includegraphics[valign=c, scale=0.15]{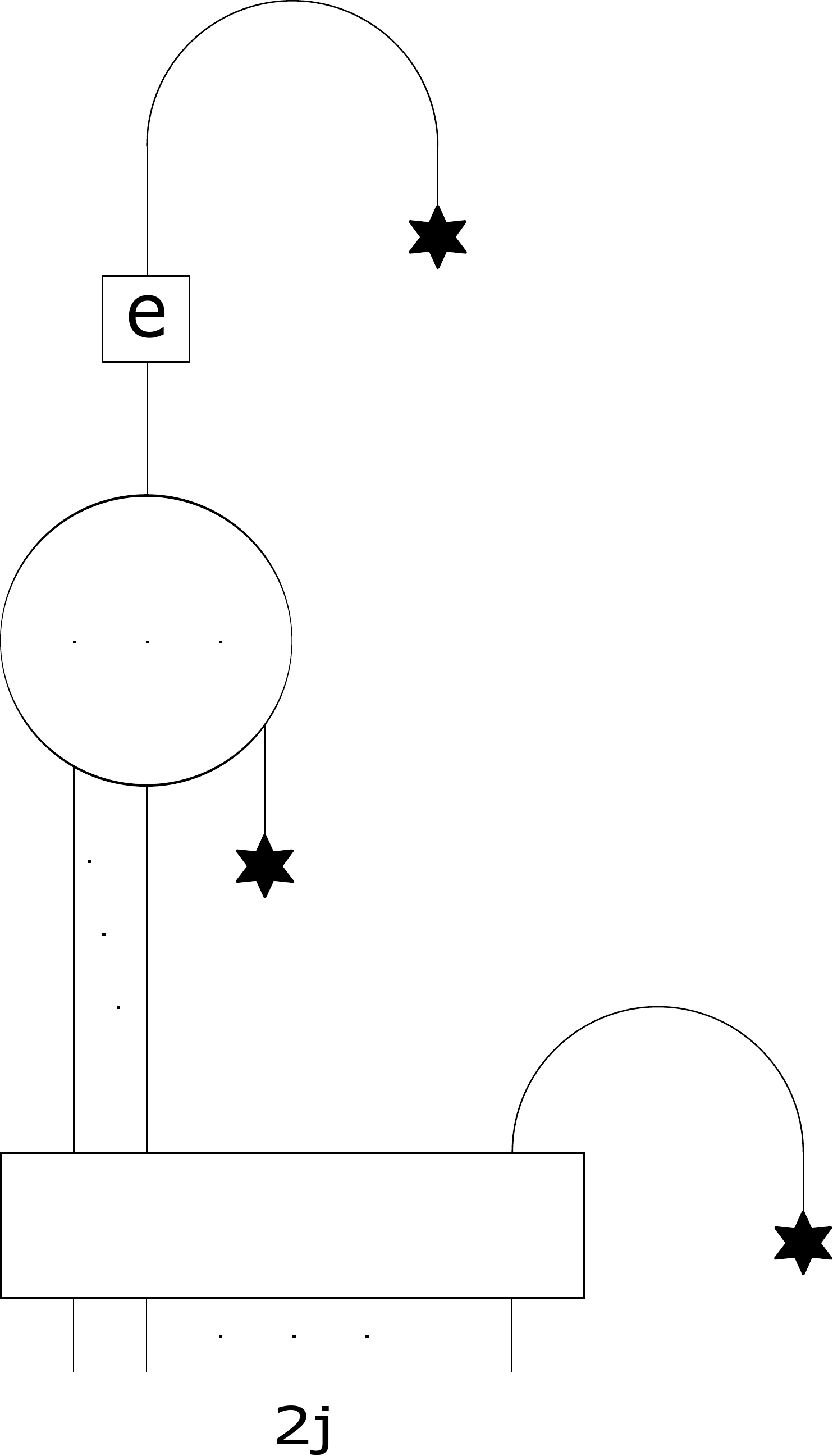} \xrightarrow{(\SE{E}{})^2} j(j+1) \quad \includegraphics[valign = c, scale = 0.15]{Spinj_down_general.pdf} \label{Spinj_eigenstate}
\end{align}
Here, the empty circle represents an arbitrary combination of fermions adding or respectively subtracting the spin of the neighbouring holonomy, 
\begin{equation}
\label{eq:updown}
    \includegraphics[valign=c, scale=0.15]{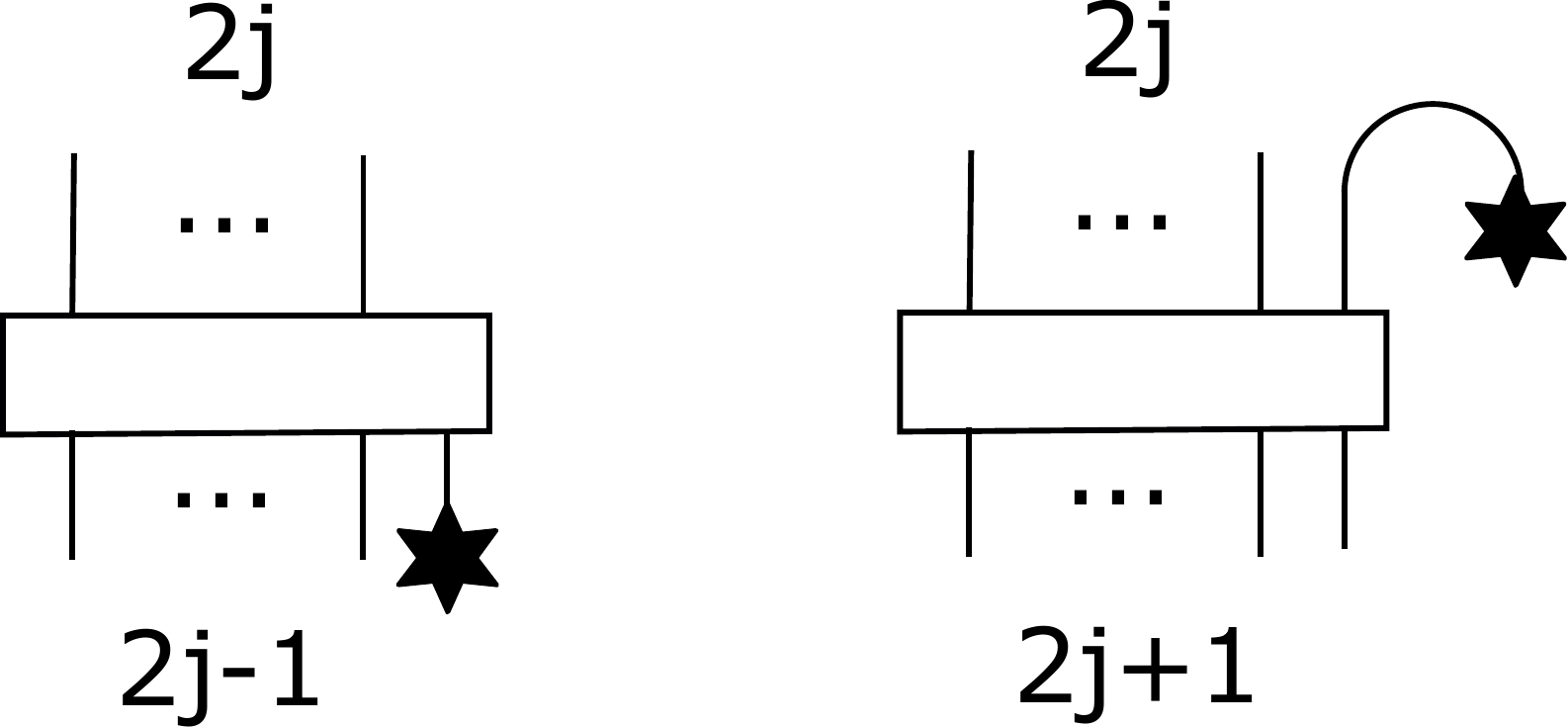}.
\end{equation}
Depending on the split of the number $n = n_+ + n_-$ of fermions into increasing and decreasing spins, the corresponding spin of the eigenstates can take values $j \in \{\frac{1}{2}, ..., \frac{n}{2}\}$ for odd $n$ and $j \in \{0, ..., \frac{n}{2}\}$ for even $n$. The eigenstates \eqref{Spinj_eigenstate} can be found by induction over $n$, but in this case a more direct route is possible.  
\begin{prop} 
\label{prop:eigen}
The total spin $\SE{E}{}$ acts on $\mathcal{H}_E^{(j)(k)}$ like a spin in the $j$-re\-pre\-sen\-ta\-tion. In particular, 
\begin{equation}
    \left.(\SE{E}{})^2\right\rvert_{\mathcal{H}_E^{(j)(k)}}=j(j+1)\one
\end{equation}
\end{prop}
\begin{proof}
The operator $\SE{e}{i}$ is just a grasping operation, it acts on the state $\jrep{\frac{1}{2}}{h_{e}}^{A}{}_{B} \otimes  \theta(s(e))^{B} \vac$ as
\begin{equation}
    \includegraphics[valign = c,scale=0.2]{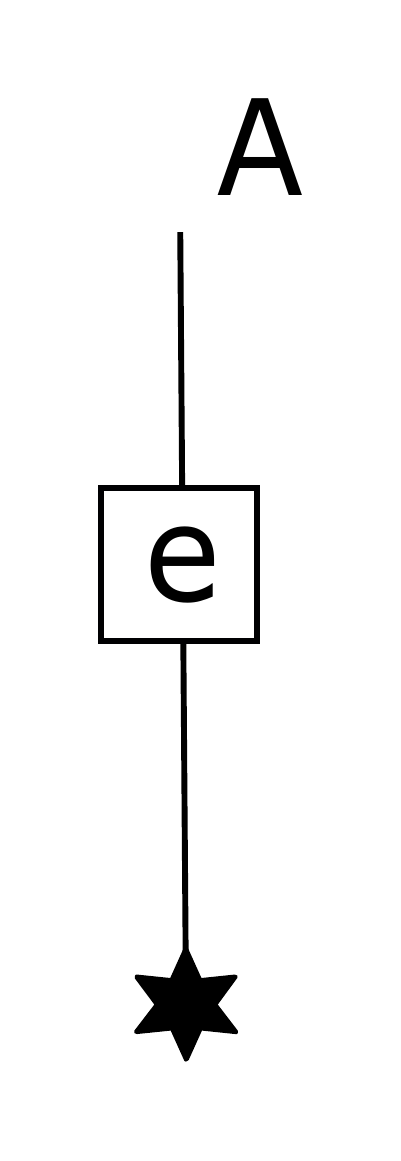}\; \overset{\SE{e}{i}}{\longmapsto}\;\includegraphics[valign = c,scale=0.2]{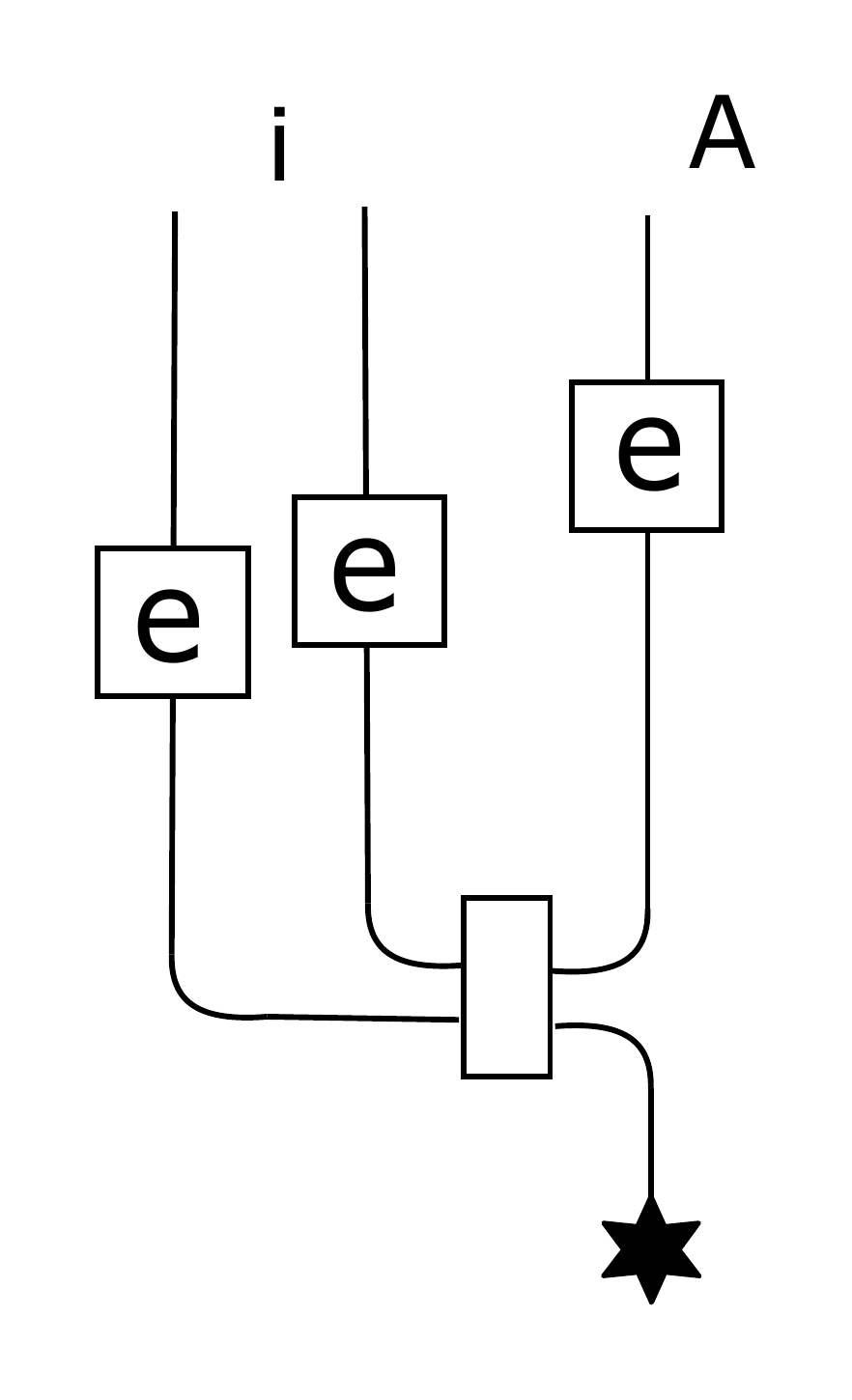}
\end{equation}
Then using the fact that we can move the intertwiner along edges, 
\begin{equation}
    \includegraphics[valign = c,scale=0.2]{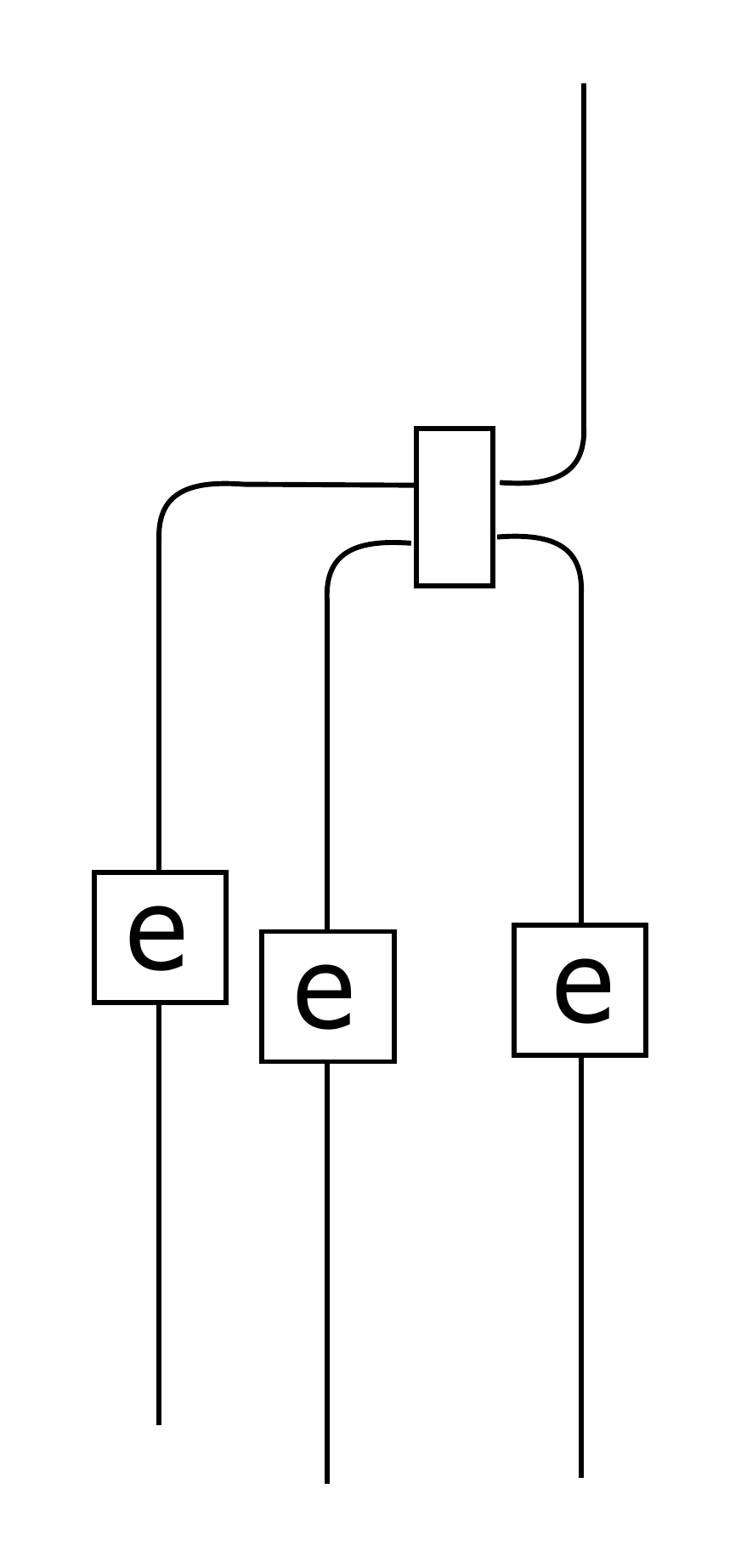}=\includegraphics[valign =c,scale=0.2]{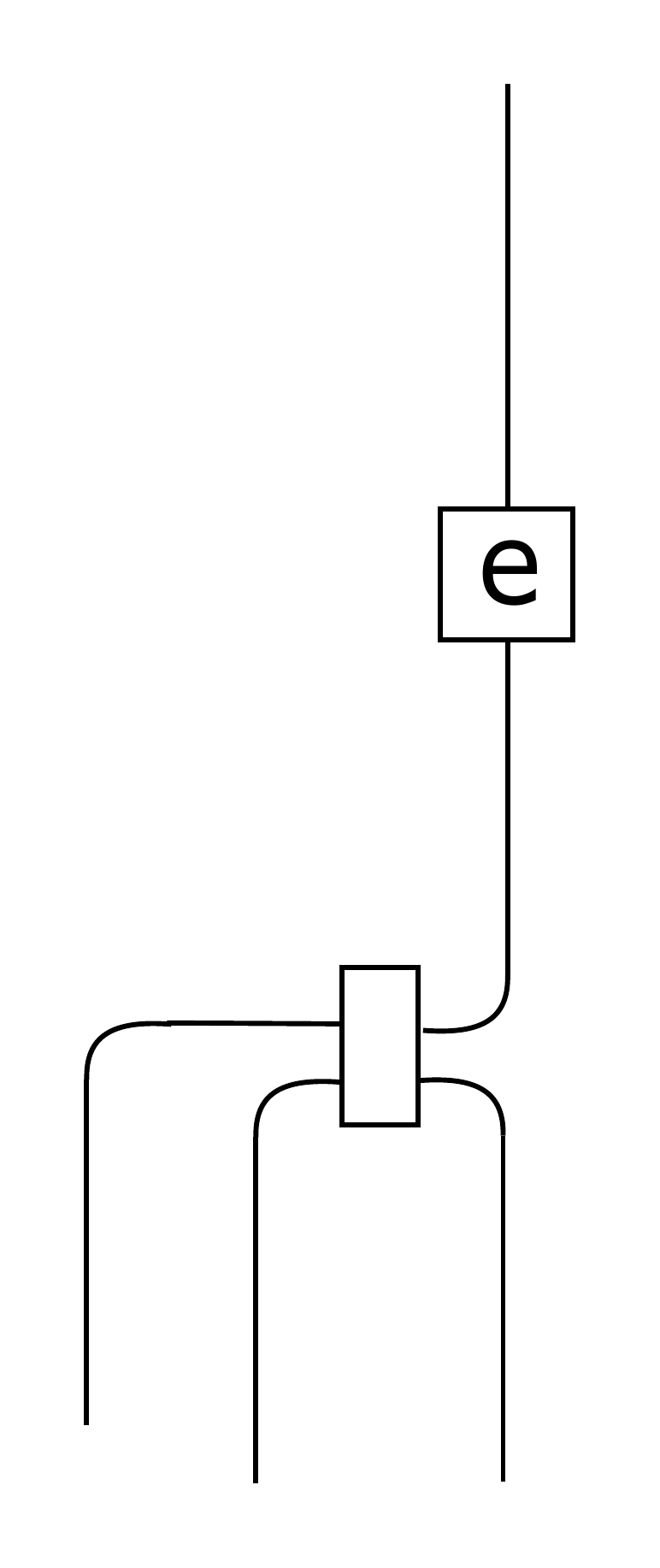}, 
\end{equation}
one sees that $\SE{e}{i}$ acts as the generator of gauge transformations at $t(E)$ on $\mathcal{H}_E$, i.e. as an angular momentum operator on 
\begin{equation}
    \bigoplus_{E'} \left((\pi_{\frac{1}{2}})^{\bigotimes|E'|}\right). 
\end{equation}
Decomposing into irreps completes the proof. 
\end{proof}
Note that the action on any state $\Psi$ in $\mathcal{H}$ that has no fermionic excitations at the positions $s(e), e\in E$ is trivial, more precisely,
\begin{equation}
    \SE{E}{}\Psi=0. 
\end{equation}
More generally, for $\{x_1,\ldots, x_m\} \cap \{s(e) | e\in E\}=\emptyset$, consider a gauge invariant operator $F$ built from arbitrary holonomies and the fermion creation operators $\theta(x_1)\ldots \theta(x_m)$, i.e., 
\begin{equation}
    F= F_{B_1\cdots B_m}[A] \otimes \theta^{B_1}(x_1)\cdots \theta^{B_m}(x_m).
\end{equation}
Then evidently
\begin{equation}
\label{eq:comm}
    \comm{\SE{E}{}}{F}=0. 
\end{equation}
Therefore 
\begin{coro}
\label{coro:eigen}
The total spin $\SE{E}{}$ acts on $F\mathcal{H}_E^{(j)(k)}$ like a spin in the $j$-re\-pre\-sen\-ta\-tion, where $F$ is any operator constructed as above. In particular, 
\begin{equation}
    \left.(\SE{E}{})^2\right\rvert_{F \mathcal{H}_E^{(j)(k)}}=j(j+1)\one
\end{equation}
\end{coro}
Corollary \ref{coro:eigen} provides for a large number of irreps for $\Se{E}{}$ and consequently a large number of eigenstates for $\Se{E}{2}$. We note however, that there are more eigenstates than those listed already. The states of corollary \ref{coro:eigen} have one point in which spin flows out. The idea to construct more general eigenstates is to consider multiple such points of outflow. Let us consider an example.   

\begin{figure}[t]
    \centering
    \includegraphics[scale= 0.1]{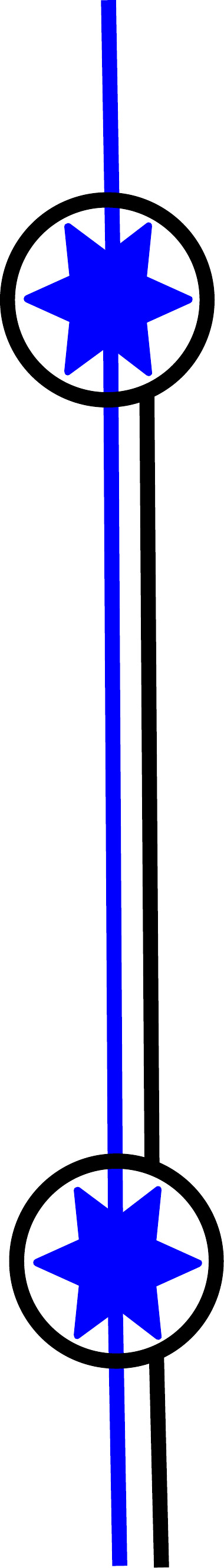}
    \caption{Sketch of a Spin Network State (blue) consisting of two fermions and a holonomy inbetween. Moreover, the two fermion system admits two points of outflow. The spin operator $\SE{E}{}$ (black) measures the parallel transported spin of the fermions at the joint endpoint, which coincides with one of the points of outflow in this case.}
    \label{fig:Singlet_operator}
\end{figure}

\begin{exa}
\label{ex:Spin1_ugly}
For two fermions and two points of outflowing spin, there is the eigenstate,
\begin{align}
    &\includegraphics[valign=c, scale=0.2]{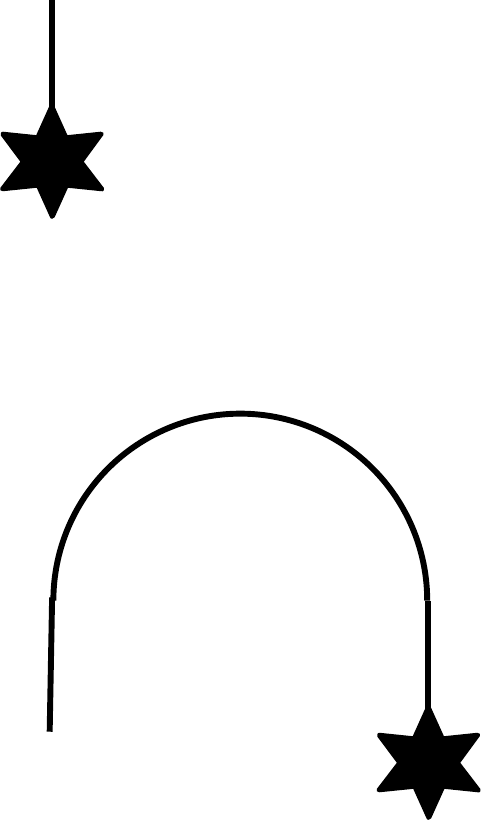} + \frac{1}{2} \quad \includegraphics[valign=c, scale=0.2]{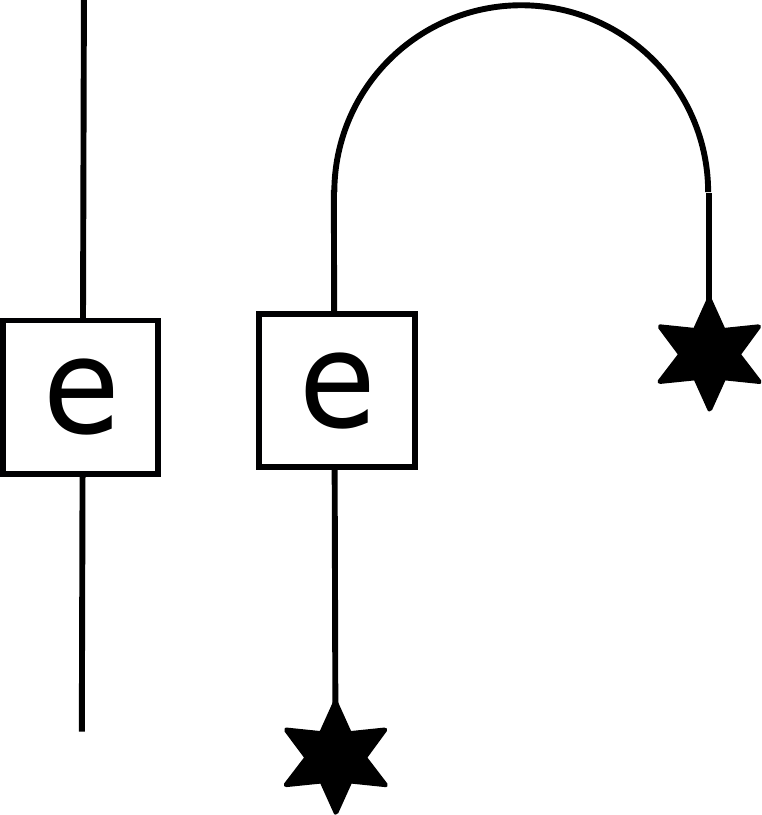} \nonumber \\
    &\xrightarrow{(\SE{E}{})^2} 2 \left( \includegraphics[valign=c, scale=0.2]{Spin1_middle.pdf} + \frac{1}{2} \quad \includegraphics[valign=c, scale=0.2]{Spin0_times_Edge.pdf} \right),
    \label{eq:triplet0}
\end{align}
which has spin 1. We notice that the first term corresponds to the new type of states while the second term is a special case of \eqref{Spinj_eigenstate} with spin $j=0$. The second term arises from the first by connecting each, the two fermions and the two points of outflowing spin, by the same holonomy $h_e$. In order to get an eigenstate, taking this linear combination is necessary, as the mixed terms \eqref{PC_2J1J2} couple the disconnected fermions. 
\end{exa}
We have to be aware of the difference between the spin network and the binor representation. Even if it is sometimes tempting to read a binor term as a spin network state or vice versa, there is a non-trivial basis change in general. The state \eqref{eq:triplet0}, for instance, is a non-trivial linear combination of spin network states and particularly not proportional to the spin network state depicted in \autoref{fig:Singlet_operator}.

In fact, we can follow the same procedure also with an arbitrary number of fermions $n$ and two points of outflow. For suitable coefficients $\lambda_{ab}, ..., \lambda_{ab ... cd}$ we can show that the generalised eigenstates of the squared total spin $(\SE{E}{})^2$ to the spin $j+k$ have the following form 
\newpage
\begin{widetext}
\begin{align}
    & \includegraphics[valign=c, scale=0.15]{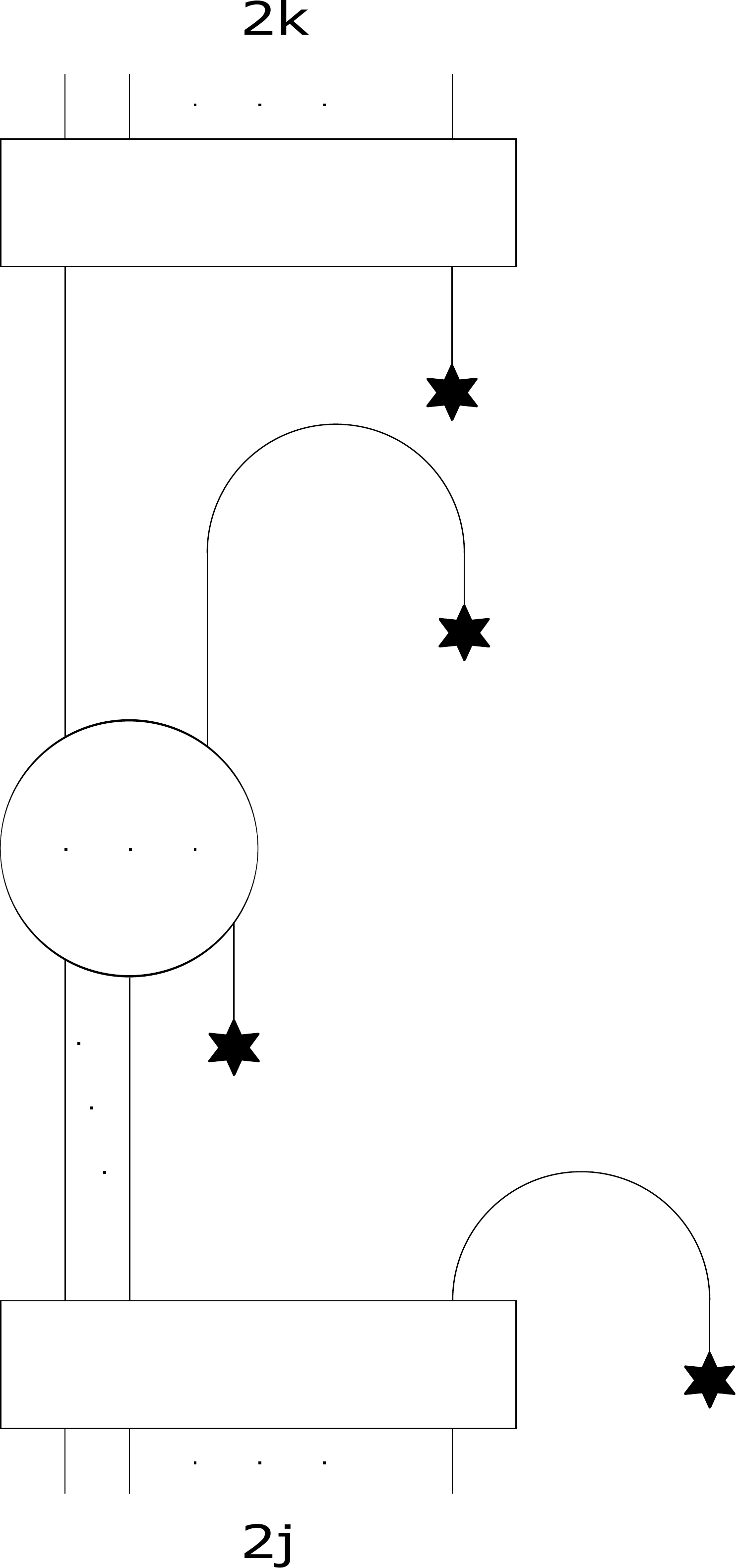} +  \sum_{a < b} \lambda_{a b} \quad \includegraphics[valign=c, scale=0.15]{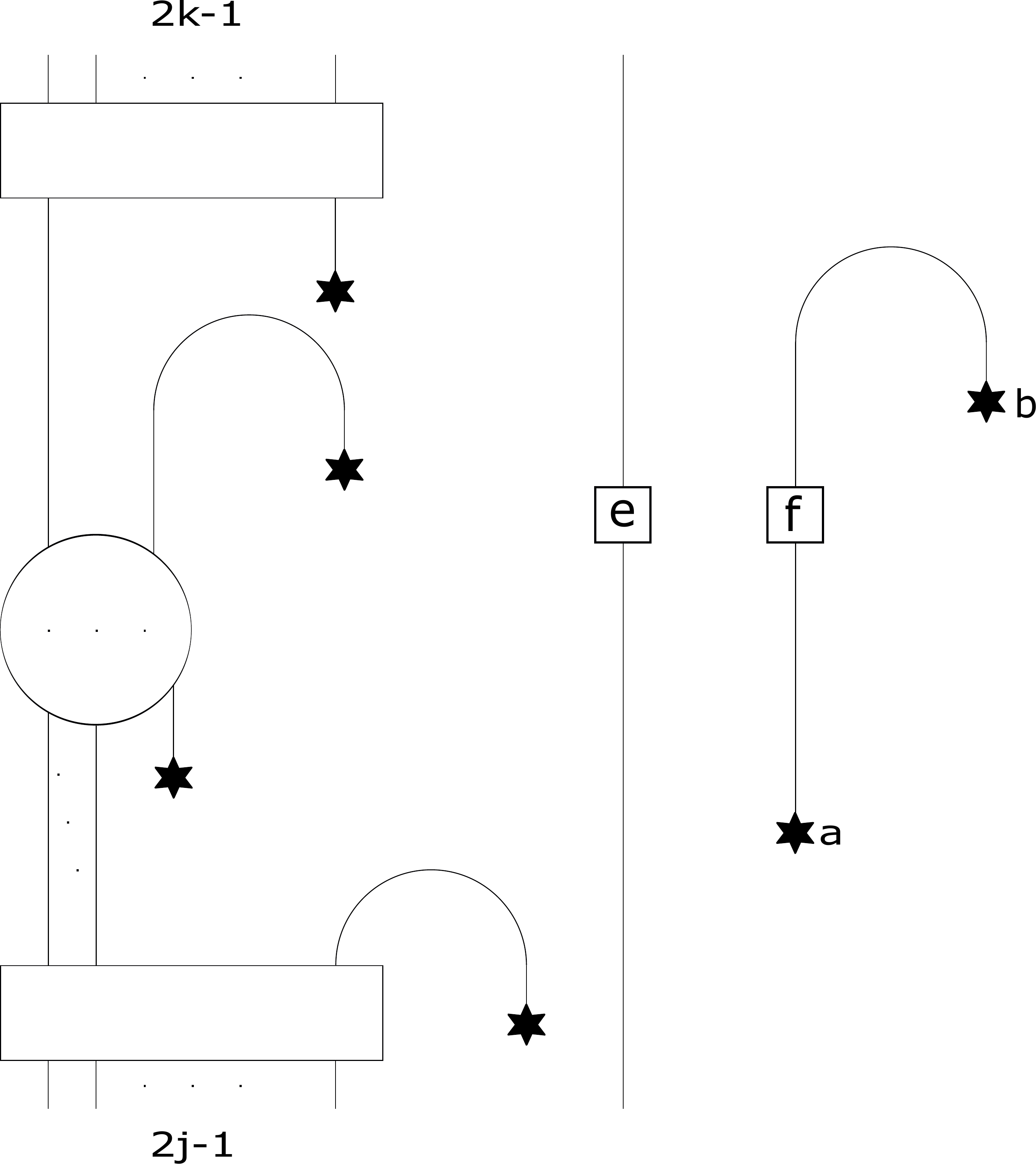} + \, ... \, + \nonumber \\[3 ex]
    &+ \sum_{a < b, ..., c < d} \lambda_{a b ... c d} \quad \includegraphics[valign=c, scale=0.15]{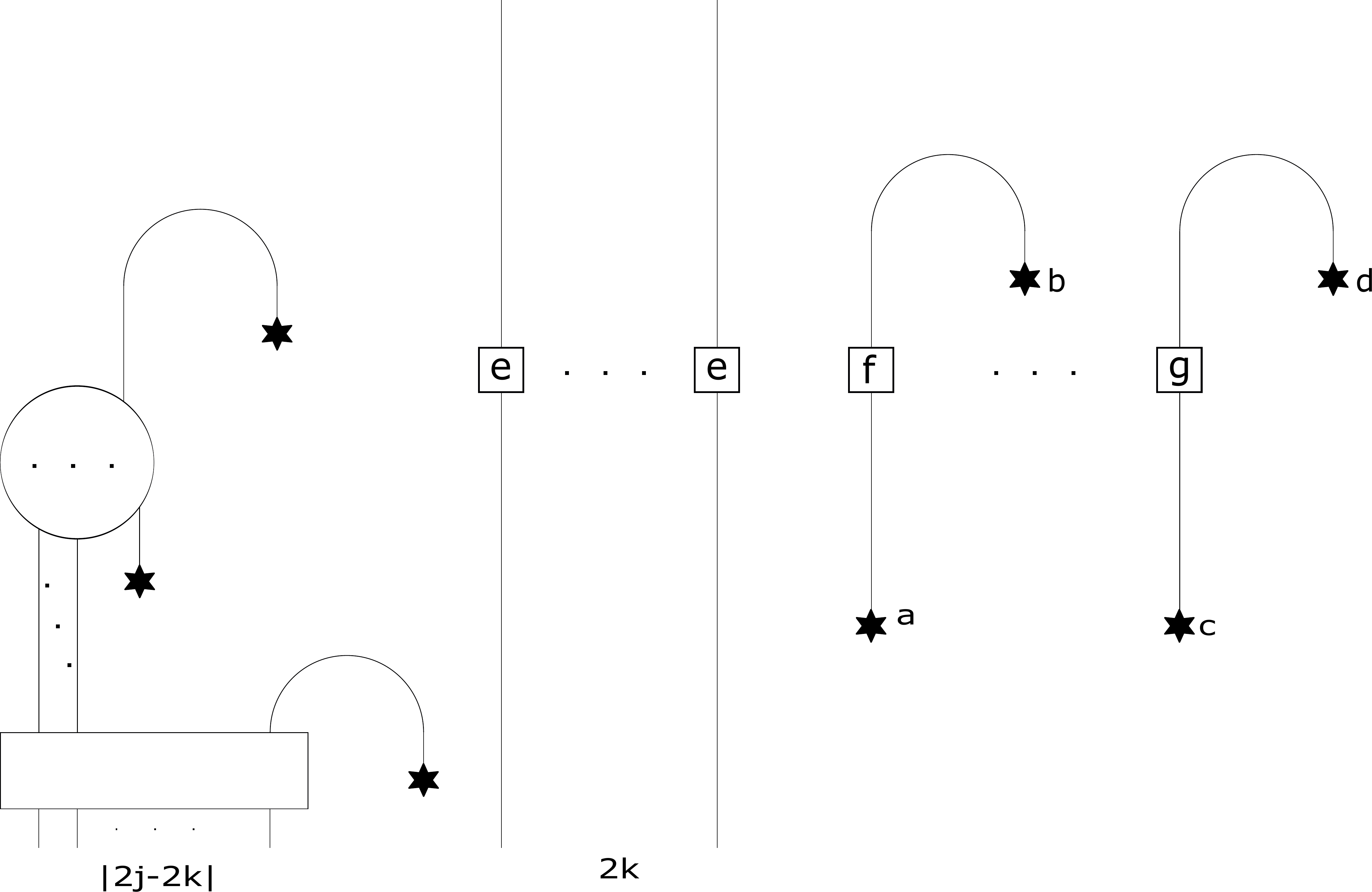} \label{Spinj+k_ugly}
\end{align}
\end{widetext}
where we numbered the fermions by the labels $a, b, ..., c, d$. Because of \eqref{eq:comm} we can even act on it with suitable operators $F$ to create further eigenvectors to the same eigenvalue. In (\ref{Spinj+k_ugly}), we assumed that there are no pure holonomies between the points of outflowing spin. From \eqref{Spinj+k_ugly} it becomes apparent that the last term is always of the form \eqref{Spinj_eigenstate}, and the sequence terminates. By acting on that state with $\SE{E}{2}$ we can read off the coefficients if we require that it is an eigenstate to the eigenvalue $j+k$.
\begin{prop}
\label{prop:ugly_eigenstates}
The states of the form (\ref{Spinj+k_ugly}) are eigenstates of $\SE{E}{2}$ to the eigenvalue $j+k$ if the coefficients are defined by the following recursion relation
\begin{align}
    \lambda^{-1}_{a_1 b_1... a_{m} b_{m}} = \lambda^{-1}_m &:= m(2j + 2k - m + 1) \lambda^{-1}_{m-1} \nonumber \\
    \lambda_0 &= 1, \label{eq:coefficients}
\end{align}
where $m$ denotes the number of reductions made. In addition $\lambda_{a_1 b_1 ... a_m b_m} = 0$ if at least one pair $(a_i, b_i)$ of fermions is of the same type (both increasing or reducing the spin from top to bottom).
\begin{proof}
See Appendix \ref{sec:proof_ugly_states}.
\end{proof}
\end{prop}

\begin{figure}[b]
    \centering
    \includegraphics[scale=0.1]{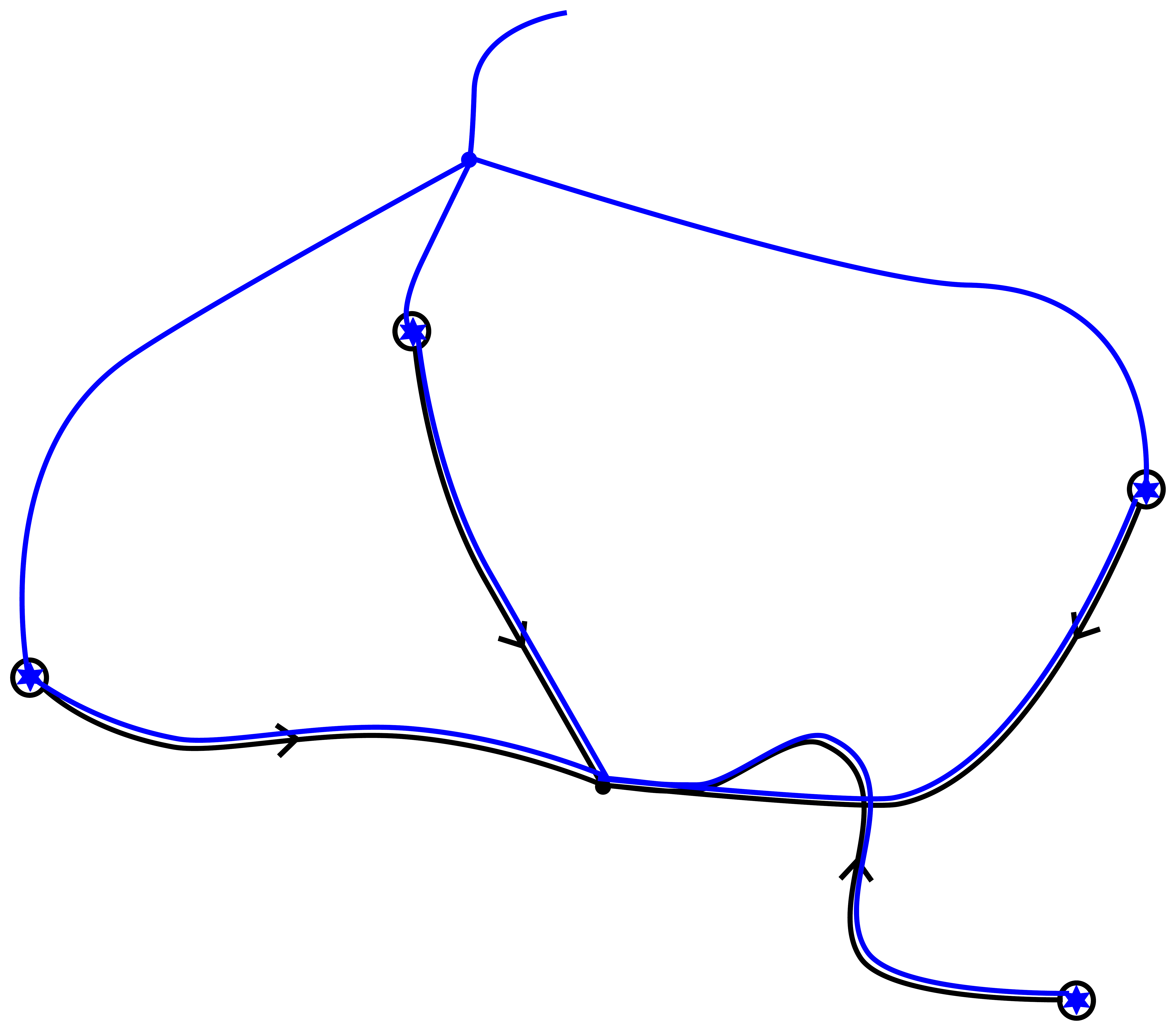}
    \caption{Sketch of a generic spin network state (blue) with four fermions and two points of outflow. The spin operator $\SE{E}{}$ (black) measures the transported spin of the fermions at the joint endpoint $t(E)$, which represents one of the two points of outflow.}
    \label{fig:generic_state}
\end{figure}

The eigenstates (\ref{Spinj+k_ugly}) can be interpreted as terms contributing to a spin network state of the form depicted in \autoref{fig:generic_state}. Note again that switching from the binor to the spin network representation or vice versa in general requires a non-trivial basis change. Proposoition \ref{prop:ugly_eigenstates} gives rise to a more general set of eigenstates. In principle, we can continue adding points of outflow. However, we would only increase the combinatorial complexity of the eigenstates as we can now reduce the outflowing spin by connecting any two of the endpoints and two fermions to another using suitable holonomies. Although the number of coefficients grows exponentially, the procedure will be exactly the same as in the case of two points of outflow.

In fact, it should be possible to write any spin network state as a linear combination of states with $n$ fermions and up to $n$ points of outflow multiplied by a spin network state with only purely gravitational degrees of freedom. This becomes apparent when writing any spin network state in the binor formalism and resolving the symmetrisations. This way, we have sketched the process of how to characterise the full spectrum of the squared total spin operators $\SE{E}{2}$.

\subsection{Projection onto a Surface Normal}
\label{sec:surfnorm}
Another observable inspired by flat standard quantum theory is the projection of the total spin onto some direction, for example  $\Se{E}{z}$. The obvious problem is that this operator is not gauge invariant and hence not a good observable. Instead, we might project the total spin onto a reference field (cf. magnetic field in a Stern-Gerlach experiment). One such reference field may be given by the normal vector field of a surface $\mathcal{S}$
\begin{equation}
    n_i:= E^a_i \epsilon_{abc} t_1^b t_2^c,  
\end{equation}
where $t_1,t_2$ are tangent vectors to a surface $\mathcal{S}$. To get rid of the density weight, we define the smeared operator

\begin{align}
    \hat{E}_i(\mathcal{S}) &= \int_{\cal{S}} \der{\mathcal{S}}_a \hat{E}_i^a = \int_\mathcal{S} \der{x}^b \wedge \der{x}^c \epsilon_{abc} \hat{E}_i^a \nonumber \\
    &:= \int_\mathcal{S} n_i \sqrt{\det\left(\leftidx{^{(2)}}{q}{}\right)},
    \label{eq:Spin_proj}
\end{align}
where $\hat{E}$ is the smeared canonical momentum operator of the Ashtekar connection $\hat{A}$ and $\leftidx{^{(2)}}{q}{}$ is the pullback of the spatial metric $q$ to the two-dimensional surface $\mathcal{S}$. The vector $n_i$ is related to the spatial surface normal by $n_a = e_a^i n_i \perp \mathcal{S}$. The operator (\ref{eq:Spin_proj}) is a measure of the area of the surface $\mathcal{S}$. This motivates the definition of a new observable 

\begin{defi}
Let $E$ be a set of edges with a joint endpoint and $\mathcal{S}$ a surface that intersects the spin network graph defined by $E$ exactly once in $t(E)$. We define the projection of the total spin onto the surface normal by
\begin{align}
    \SprojE = \SE{E}{i} \hat{E}_i (\mathcal{S}).
    \label{eq:def_JE}
\end{align}
\end{defi}
From the fact that $\hat{E}_i$ couples also on pure gravitational spin network functions, one can see that it will not have the same eigenbasis as $\SE{E}{2}$. This can be also proven by explicitly calculating the commutator. Already for $n=2$ fermion spins, we get a non-trivial contribution of the commutator of $\hat{h}_e$ and $\hat{E}_i$,
\begin{widetext}
\begin{align}
    \left[ \SE{E}{2}, \SprojE \right] = &\Big( \Sx{e_1(0)}{l} \tensor{\left( h_{e_1}^{-1} \right)}{^m_l} \Sxl{e_1(0)}{j} \Sx{e_2(0)}{k} \nonumber \\
    &+ \Sx{e_2(0)}{l} \tensor{\left( h_{e_2}^{-1} \right)}{^m_l} \Sxl{e_1(0)}{j} \Sx{e_2(0)}{k} \Big) \tensor{\left( h_{e_1} \frac{\tau_m}{2} h_{e_2}^{-1} \right)}{^j_k} \neq 0.
\end{align}
\end{widetext}
On the other hand, the two operators in (\ref{eq:def_JE}) do commute.
\begin{align}
    \left[ \SE{E}{i}, \hat{E}_i (\mathcal{S}) \right] = \sum_{e \in E} \tensor{\tau}{_i^i_j} \tensor{\jrep{1}{h_{e}^{-1}}}{^j_k} S(s(e))^k = 0,
    \label{eq:commutator_SprojE}
\end{align}
which is a result of the total antisymmetry of $\tau$. As a consequence, the self-adjointness of (\ref{eq:def_JE}) is ensured and we do not have to think about the order of the two operators.

Although there is no hope to get a spin component with exactly the same properties as $\SE{z}{}$ with (\ref{eq:def_JE}), we are able to understand some special cases of spin coupling. As $\hat{E}_i$ acts with the insertion of a Pauli matrix $\tau_i$, we find that every addend of (\ref{eq:def_JE}) acts very similar to (\ref{PC_2J1J2}). In particular, the binor representation is, up to a multiplicative factor, identical with the operator being inserted at the intersection point of $\mathcal{S}$ with $e$ and at the point where the respective fermion is located. The factor is given by the defining relation, which is known from the area operator in loop quantum gravity \cite{Rovelli:1994ge},
\begin{align}
    \hat{E}_i \hat{E}^i h_e = -\left(8 \pi l_p^2 \gamma\right)^2 h_{e_1} \tau_i \tau^i h_{e_2} = \left(8 \pi l_p^2 \gamma\right)^2 j(j+1) h_e,
    \label{eq:area_squared}
\end{align}
with $e = e_1 \circ e_2$ being split at the intersection point with $\mathcal{S}$. Note that we assume now and for the following discussion that the edge $e$ intersects the surface $\mathcal{S}$ exactly once and transversely, i.e. it punctures the surface non-tangentially. Furthermore, the intersection point equals neither $e(0)$, nor $e(1)$.

In order for (\ref{eq:area_squared}) to hold, we have to multiply the binor representation (\ref{PC_2J1J2}) with the spin $j$ of the holonomy $h_e$, which intersects $\mathcal{S}$ and a constant factor $\left(8 \pi l_p^2 \gamma\right)$.

Consider now two fermions of the form (\ref{Spinj_eigenstate}) with the spin $j = 0$. This kind of state involving two fermions was first considered in \cite{Morales1995}. If we admit an intersection of $\mathcal{S}$ with an arbitrary point on the edge, we can calculate the action of the spin projection,
\begin{align}
    \includegraphics[valign=c, scale=0.19]{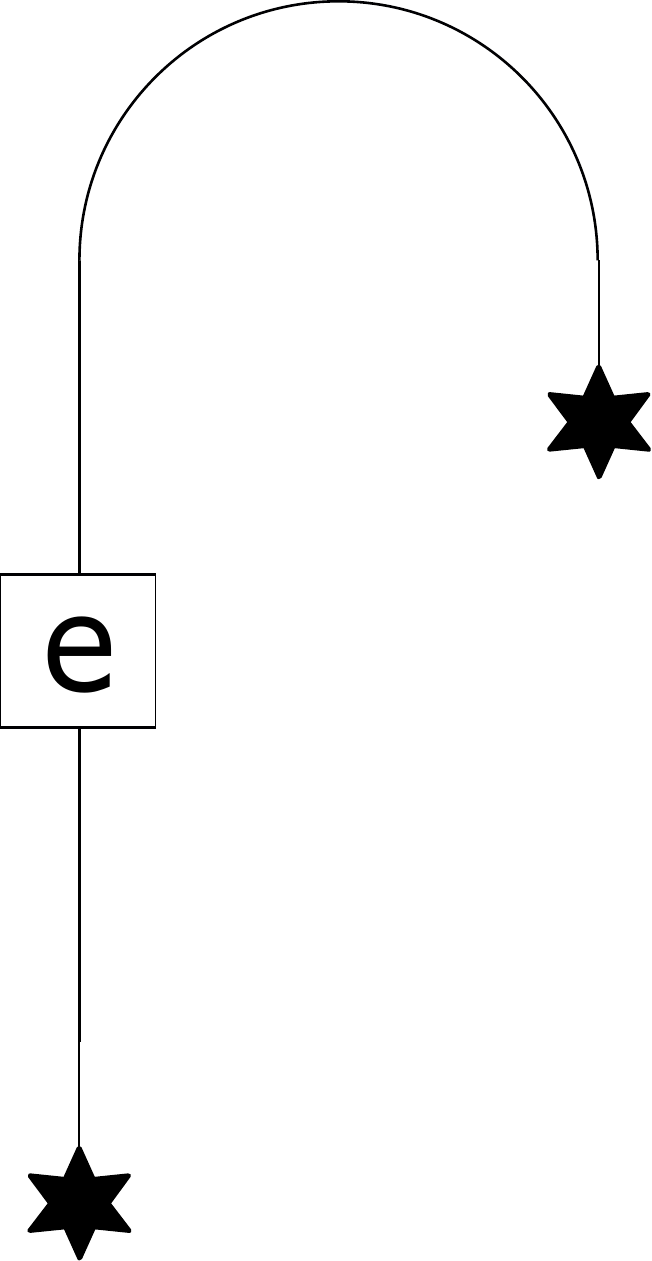} \xrightarrow{\SprojE} \frac{1}{2} \left[ \includegraphics[valign=c, scale=0.19]{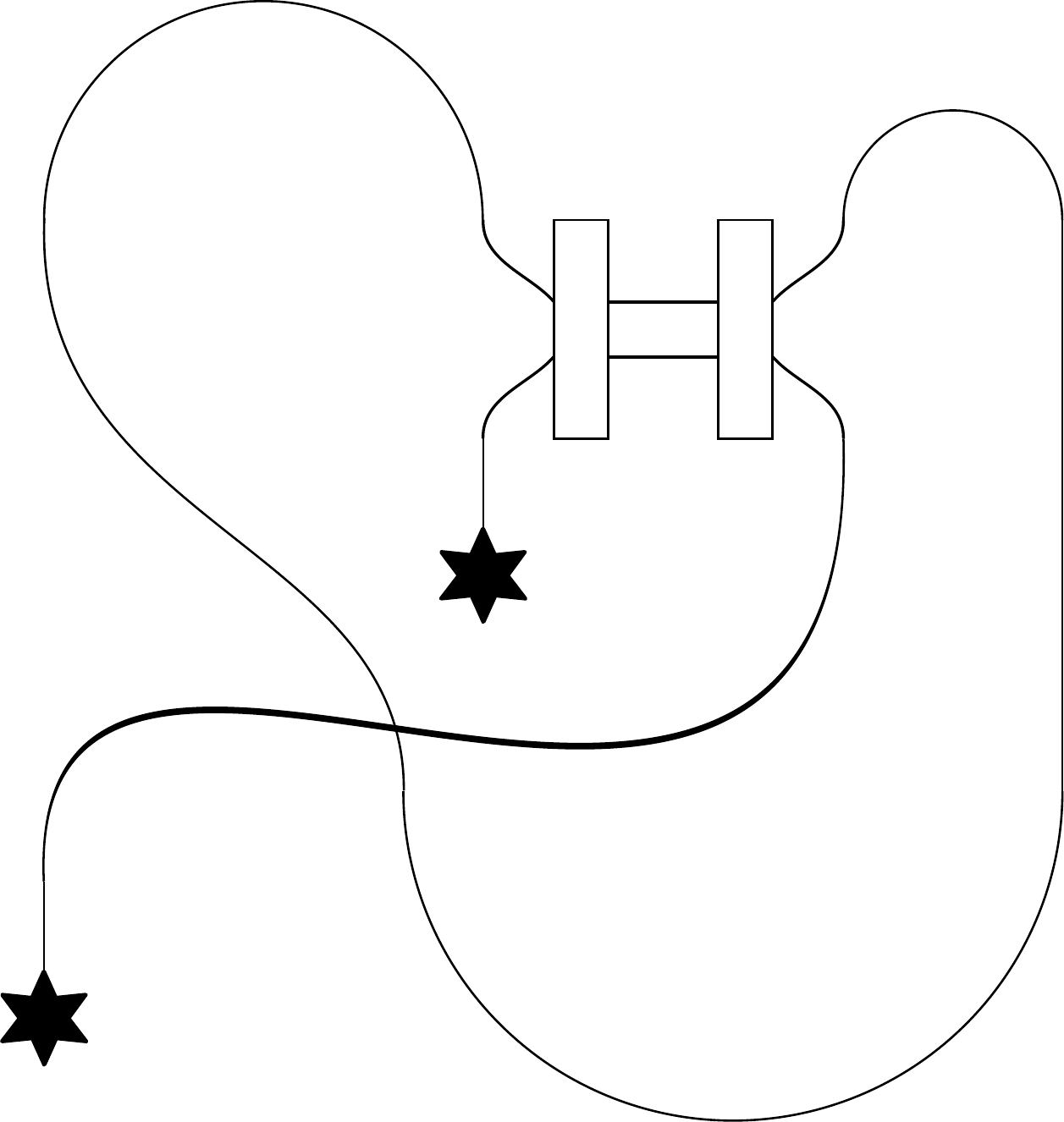} - \includegraphics[valign=c, scale=0.19]{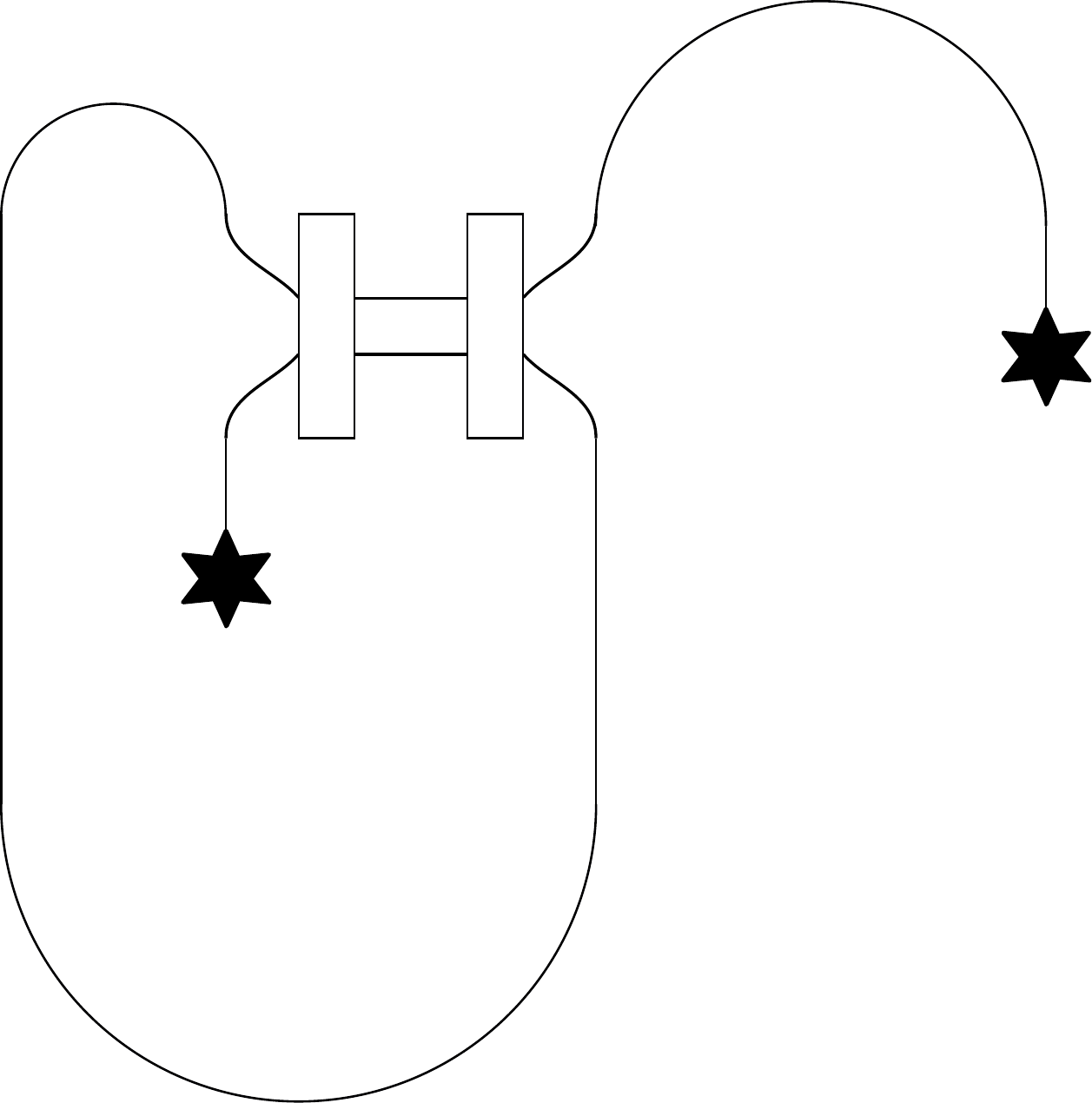} \right] = 0,
    \label{eq:singlet_project}
\end{align}
where we again omitted the "e" label in the calculation. (\ref{eq:singlet_project}) is straight-forwardly calculated using the basic techniques of binor calculus. This result would be expected for a component of a vanishing spin projection. More generally, we can also consider the product of the singlet (\ref{eq:singlet_project}) with a purely gravitational spin network. Here, we will use the Leibnitz rule of $\hat{E}_i$ and the term where $\SprojE$ acts on the singlet vanishes. The latter is vanishes by (\ref{eq:singlet_project}) and the first term reads
\begin{widetext}
\begin{align}
    \includegraphics[valign=c, scale=0.5]{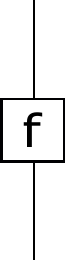}_{2j} \cdot \includegraphics[valign=c, scale=0.15]{Spin0.pdf} \xrightarrow{\SprojE} \frac{1}{2} \left[ \leftidx{_{2j}}{\includegraphics[valign=c, scale=0.1]{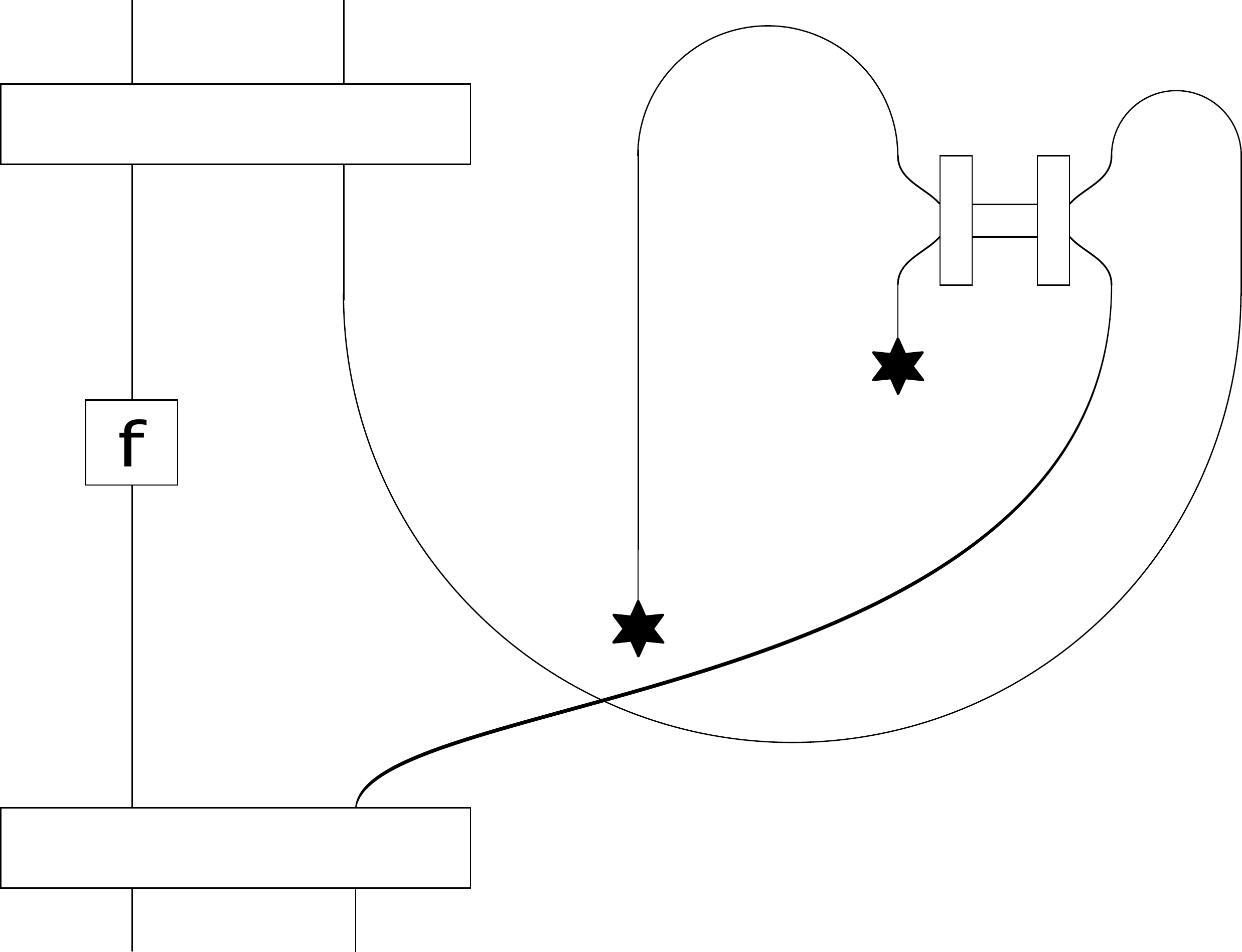}}{} + \leftidx{_{2j}}{\includegraphics[valign=c, scale=0.1]{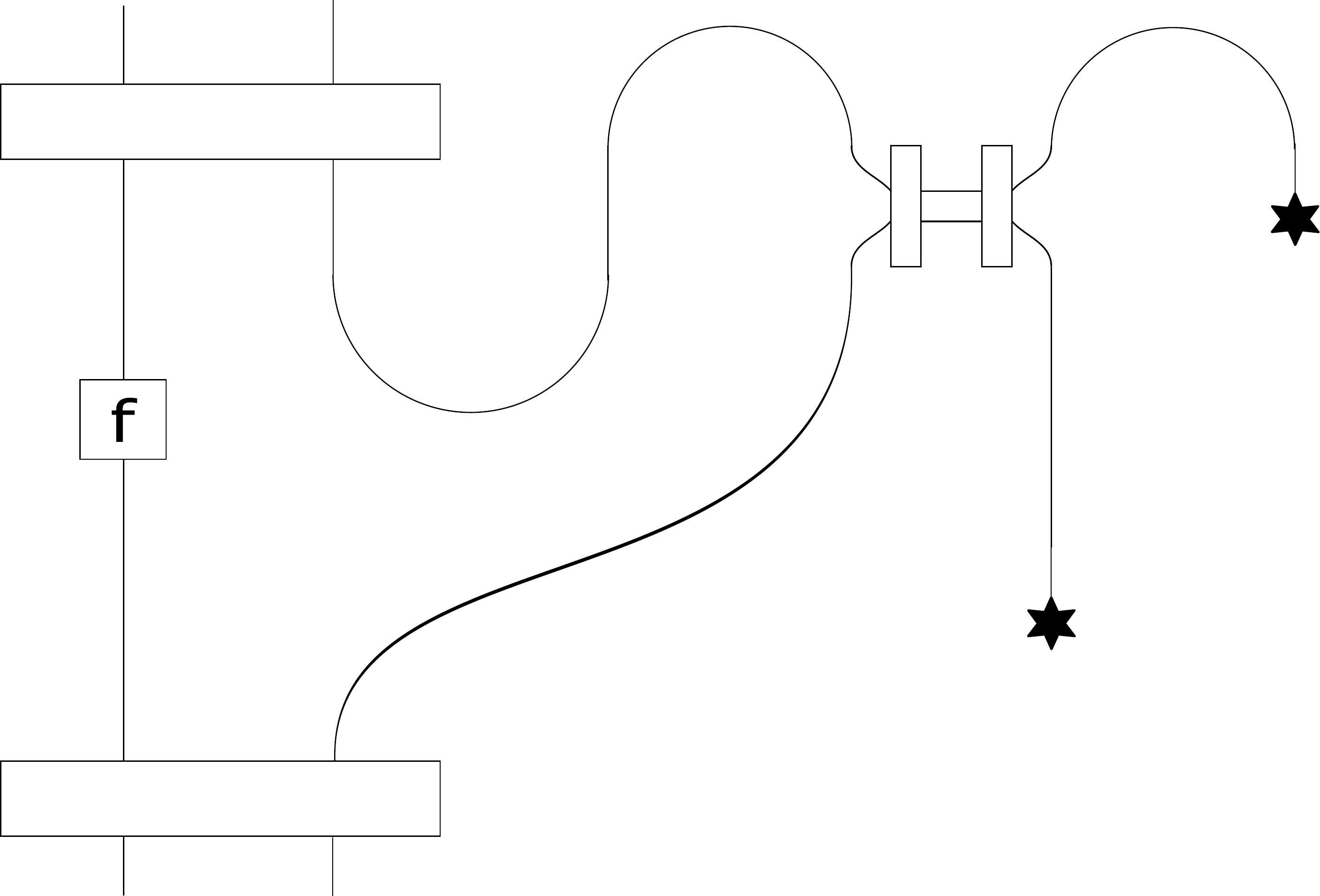}}{} \right] = 0.
\end{align}
\end{widetext}
This can be calculated by dissolving the symmetrisations, pulling two holonomies across each other and consequently using the binor identity. Note that we assumed the edge $f$ to intersect the surface $\mathcal{S}$ at the same point as $e$. Apart from that, $f$ is arbitrary. If we allow $f$ to intersect $\mathcal{S}$ at any arbitrary point, the projection operator $\SprojE$ will also map the state to zero and the calculation is similar to (\ref{eq:singlet_project}). The freedom to multiply any edge $f$ here is a result of the symmetry of the state. In the following however, we have to stick to $f=e$. Analogously, one can show that also the spin 1 states are eigenstates of (\ref{eq:def_JE}).

\begin{align}
    \includegraphics[valign=c, scale=0.6]{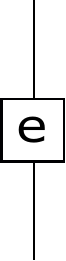}_{2j} \cdot \includegraphics[valign=c, scale=0.15]{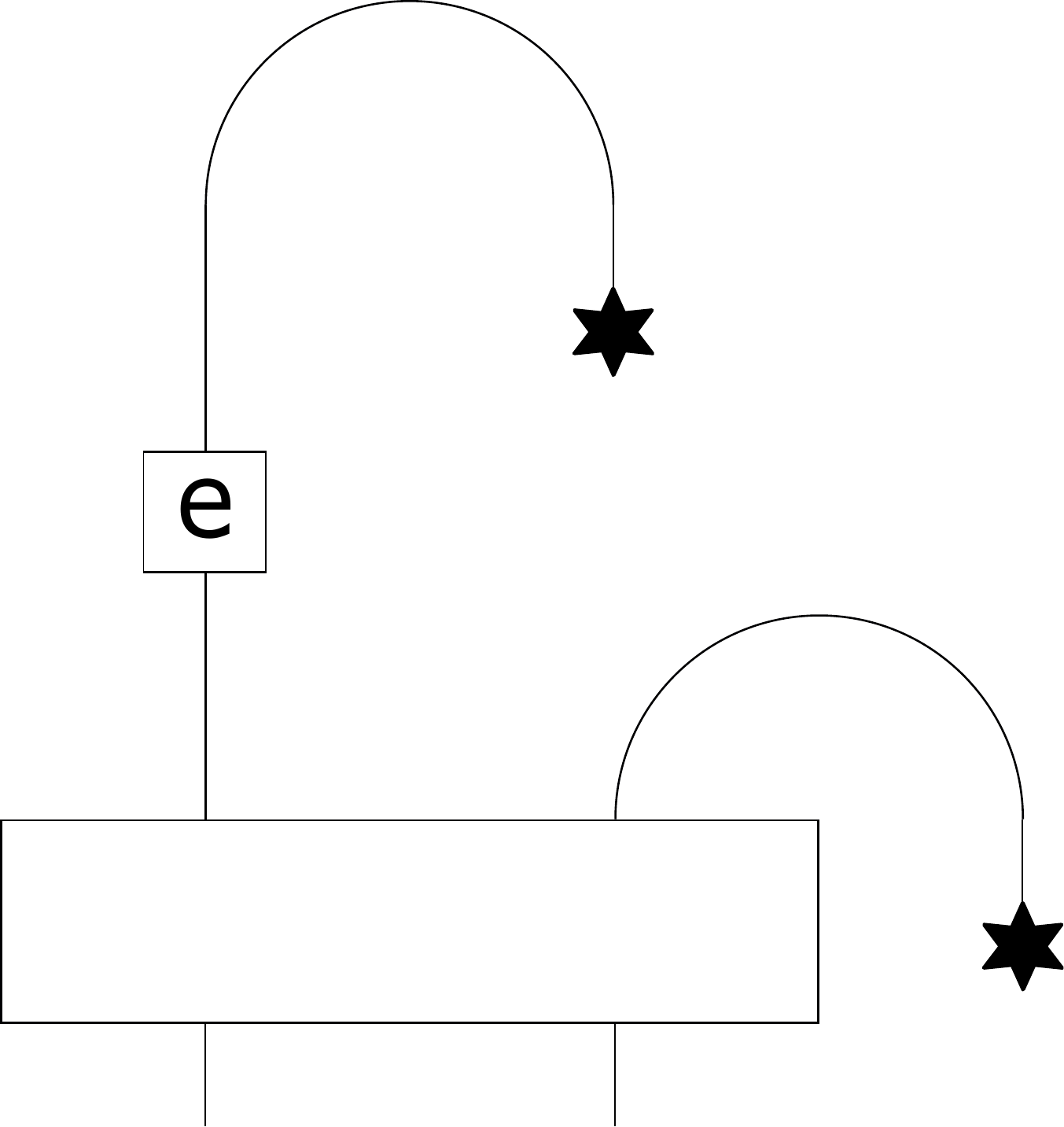} \xrightarrow{\SprojE} - \frac{1}{2} (j+1) \quad \includegraphics[valign=c, scale=0.6]{Edge_withSU2_rep.pdf}_{2j} \cdot \includegraphics[valign=c, scale=0.15]{Spin1_down.pdf}
    \label{eq:triplet_projectdown}
\end{align}

\begin{align}
    \includegraphics[valign=c, scale=0.6]{Edge_withSU2_rep.pdf}_{2j} \cdot \includegraphics[valign=c, scale=0.18]{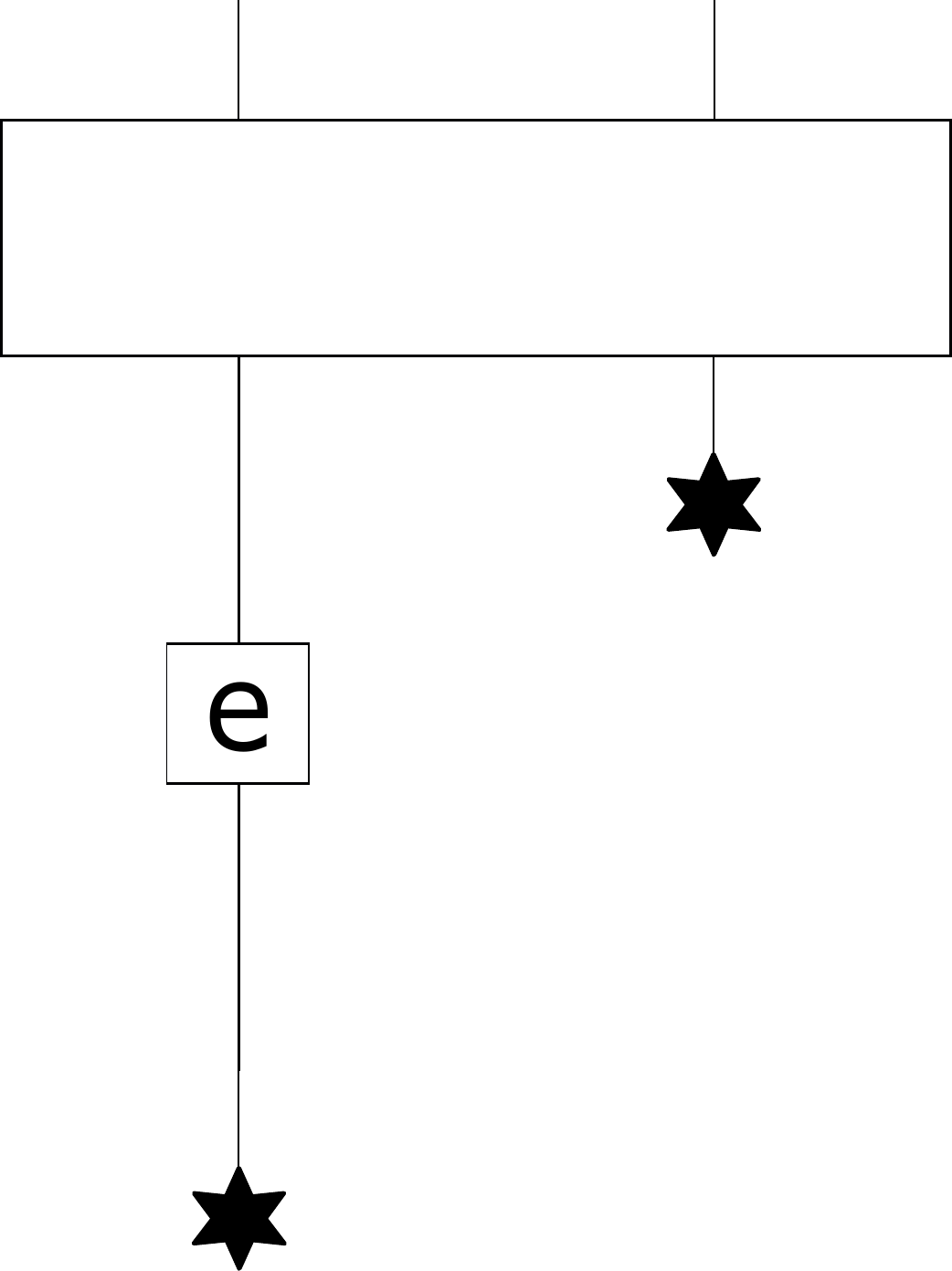} \xrightarrow{\SprojE} \frac{1}{2} (j+1) \quad \includegraphics[valign=c, scale=0.6]{Edge_withSU2_rep.pdf}_{2j} \cdot \includegraphics[valign=c, scale=0.18]{Spin1_up.pdf}_.
    \label{eq:triplet_projectup}
\end{align}
Since it is always the same techniques used in the calculations, we refrain from writing out all the steps. The structure of the eigenvalues is reminiscent of the projection of the "up-up" and "down-down" triplet states from flat quantum theory, but now the quantisation axis is defined by the orthonormal of the surface $\mathcal{S}$. Indeed, if we flip the sign of the surface normal, also the signs of (\ref{eq:triplet_projectdown}) and (\ref{eq:triplet_projectup}) flip, because the operator $\hat{E}$ includes the sign of the inner product of the surface normal and tangent vector of $e$ at the intersection point. 

Obviously, the operator $\hat{E}_i$ influences the actual eigenvalue, as it measures the spin $j+ \frac{1}{2}$ of the holonomies $h_e$, whose edges intersect the surface $\mathcal{S}$. If we were to normalise the surface normal in the definition of the projection operator (\ref{eq:def_JE}) such that 
\begin{align}
    \tr{\left( \hat{E}^*_i \hat{E}^i \right)} = 1,
\end{align}
we will gather an additional factor $\frac{2}{\sqrt{j(j+1)}}$ for $\hat{E}^i \propto \frac{1}{2} \tau_i$ in the spin $j$ representation\footnote{We would arrive at the same result, if we would define the normalised operator by dividing $\SprojE$ by the surface area $A_\mathcal{S}$.}. Since in (\ref{eq:singlet_project} - \ref{eq:triplet_projectup}), we are acting on a spin $j + \frac{1}{2}$ holonomy, the eigenvalues would take the form $\pm \frac{j+1}{\sqrt{\left(j+\frac{1}{2} \right)\left(j + \frac{3}{2}\right)}}$. If we take the limit $j \to \infty$, these reproduce the spin projection $\pm 1$, which we would expect for the triplet state in flat quantum theory. Indeed, in the literature \cite{Mikovic:2008zz, Mikovic:2011zx} the limit of large spin is often regarded as a semiclassical limit of loop quantum gravity. The third spin 1 state (\ref{eq:triplet0}) completes the discussion of the spin projection operator (\ref{eq:def_JE}) on a system of two fermions. One can show in an analogous manner that
\begin{align}
    \includegraphics[valign=c, scale=0.2]{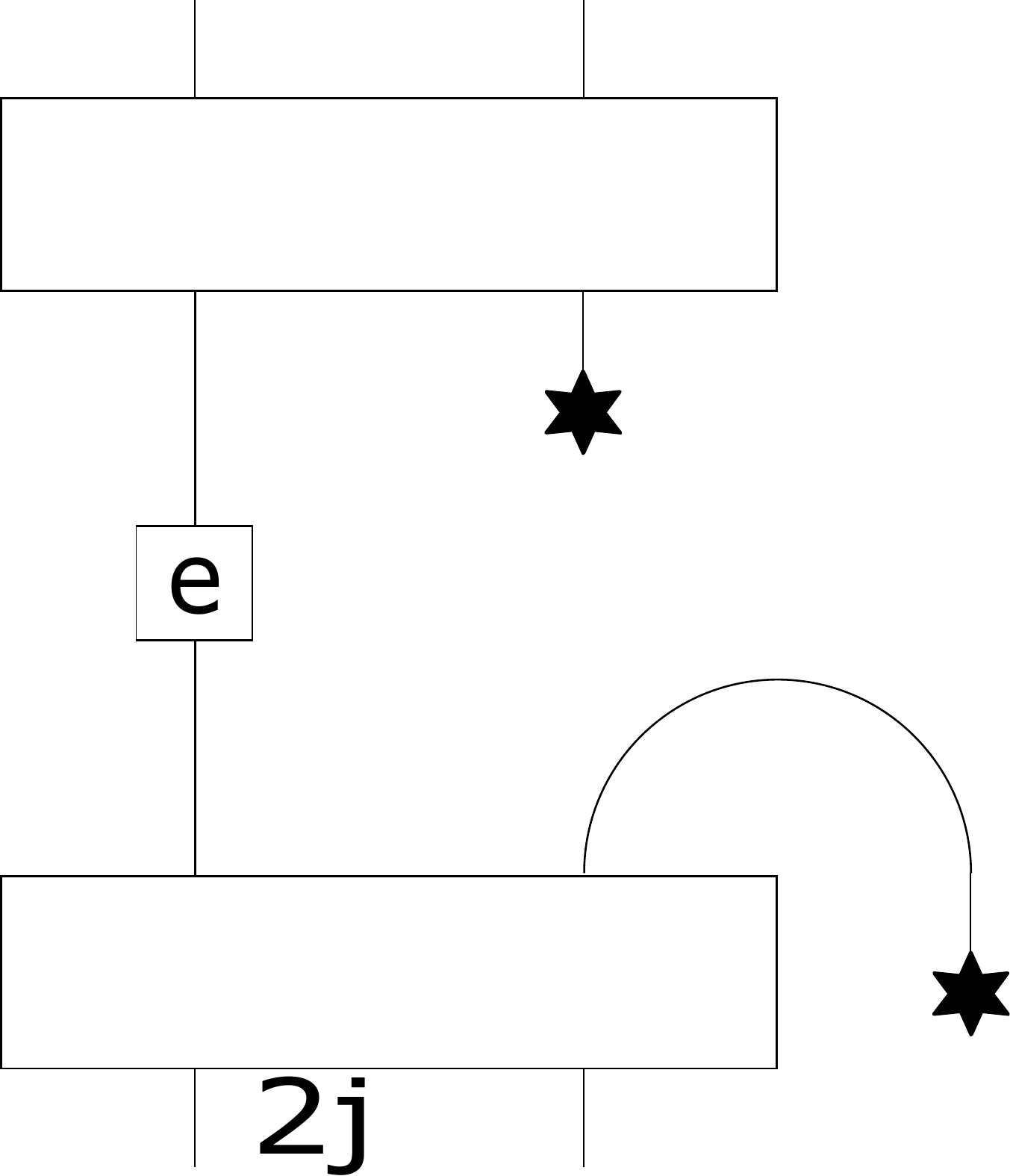} \xrightarrow{\SprojE} 0. 
    \label{eq:triplet_projectupdown}
\end{align}
holds. With (\ref{eq:singlet_project}) and (\ref{eq:triplet_projectdown}) we now understand both terms in (\ref{eq:triplet0}) in the context of the spin projection. Hence, we find that the third triplet state besides (\ref{eq:triplet_projectdown}) and (\ref{eq:triplet_projectup}) is represented by (\ref{eq:triplet0}), which vanishes under the action of $\SprojE$. Moreover, we have found one singlet state (\ref{eq:singlet_project}) in complete analogy to flat quantum theory. 

Note that it is possible to generalise the above statement to the coupling of $n$ fermions in principle. However, we cannot change the purely gravitational spin network state to be different from a sole holonomy $h_e$, otherwise it will not be an eigenstate of (\ref{eq:def_JE}). One counter example is
\begin{exa}
Consider the two fermion state with an intertwiner inbetween
\begin{align}
    \includegraphics[valign=c, scale=0.2]{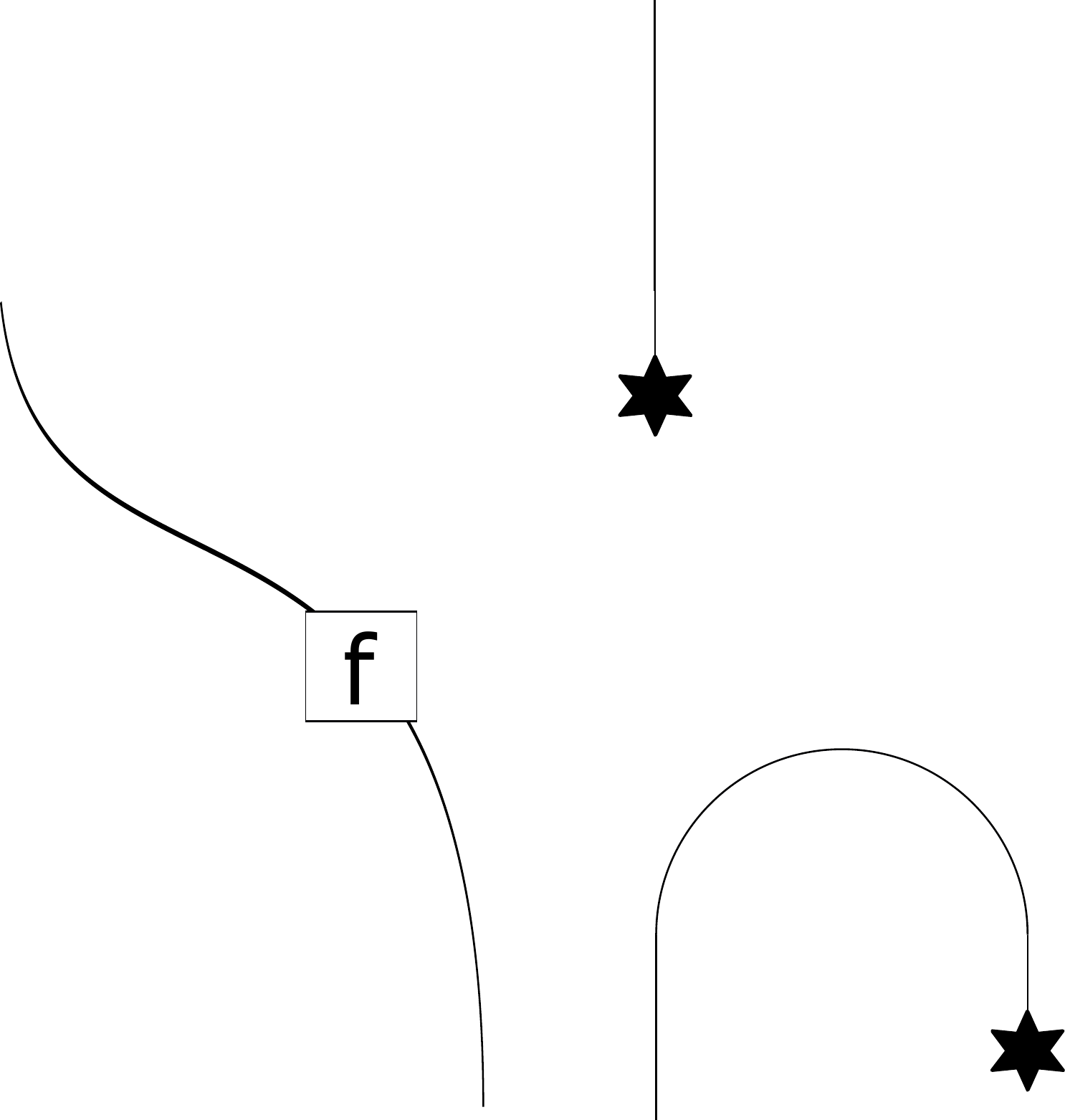} = \quad \includegraphics[valign=c, scale=0.2]{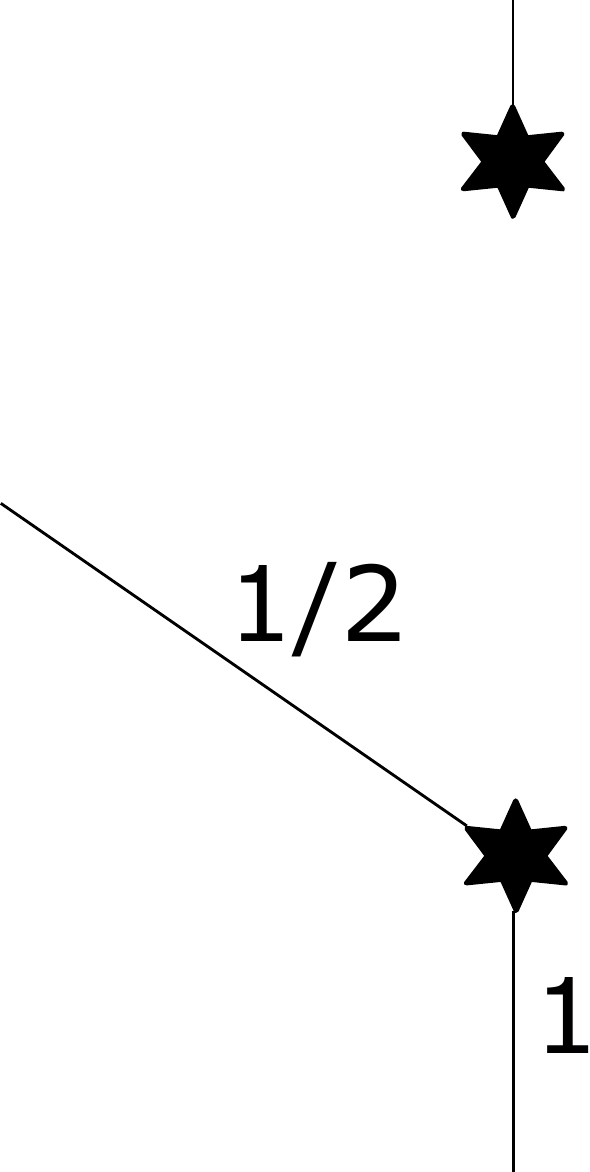}_,
\end{align}
where the right hand side represents the spin network of the state which is written in the binor formalism on the left. We only inserted one non-trivial intertwiner. When applying the same techniques as above, we end up with a different state
\begin{align}
    \includegraphics[valign=c, scale=0.2]{JE_counter.pdf} \xrightarrow{\SprojE} \frac{1}{2} \quad \includegraphics[valign=c, scale=0.2]{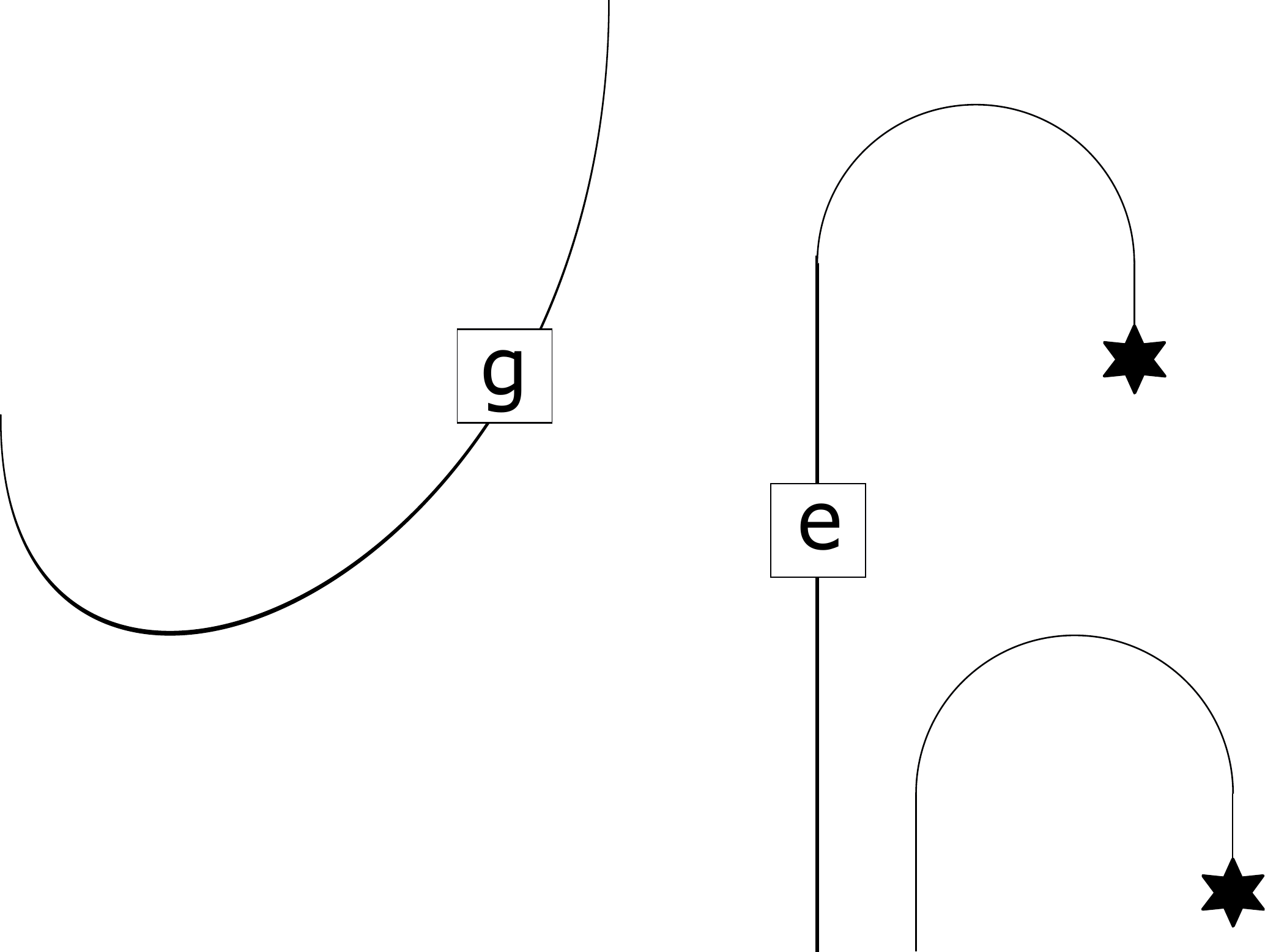}_,
    \label{eq:JE_count}
\end{align}
where we constructed the counter example such that the intersection point of $S$ with the spin network graph lies on the edge $f$ and $e = f \circ g^{-1}$.
\end{exa}
We showed that the projection (\ref{eq:def_JE}) can in principle mimic the observable $\SprojE$ in very special cases. In general, however, there is a non-trivial interaction between $\hat{E}_i$ (and, in consequence $\SprojE$) and gravity. As a consequence, $\SE{E}{2}$ and $\SprojE$ do not have the same eigenbasis.

\subsection{Projection onto the Spin of a Specific Fermion}
In this section, we want to go away from a manifestly gravitational reference frame and draw attention to another observable measuring the projection of $\SE{E}{i}$ onto the spin of one specific fermion of the fermion system.
\begin{defi}
In the case of $n$ fermions connected by a set of edges $E'$ with a joint endpoint $t(E')$ and a $n+1^\text{st}$ fermion forming $E = E' \cup \{e_{n+1}\}$, we define the observable
\begin{align}
    \SprojF &= 2 \SE{E'}{i} \, S_{\{ e_{n+1} \} \, i}.
    \label{eq:def_Jn+1}
\end{align}
\end{defi}
Note that we exclude the self-projection $S_{\{ e_{n+1} \}}^{2}$ in (\ref{eq:def_Jn+1}) by projection onto the spin defined by $E'$ instead of $E$. Again, the addends of the observable (\ref{eq:def_Jn+1}) have the same action as (\ref{PC_2J1J2}) with being inserted at the points where the respective fermions are located. Let us have a look at the action of (\ref{eq:def_Jn+1}) on the eigenstates of $\SE{E}{2}$ discussed in (\ref{Spinj_eigenstate}). 
\begin{prop}
The eigenstates (\ref{Spinj_eigenstate}) of $\SE{E}{2}$ are also eigenstates of $\SprojF$ to the eigenvalue $j$, the spin of the state corresponding to the eigenvalue $j(j+1)$ of $\SE{E}{2}$, i.e.
\begin{align}
    \psi \cdot \includegraphics[valign=c, scale=0.12]{Spinj_down_general.pdf}_{n+1} &\xrightarrow{\SprojF} j \cdot \psi \cdot \includegraphics[valign=c, scale=0.12]{Spinj_down_general.pdf}_{n+1} \label{eq:Jn+1_projectup}
\end{align}
If the $n+1^\text{st}$ fermion is of the type which reduces the total spin of the $n$ fermion system, the eigenvalue reads $-(j+1)$, i.e.
\begin{align}
    \psi \cdot \includegraphics[valign=c, scale=0.12]{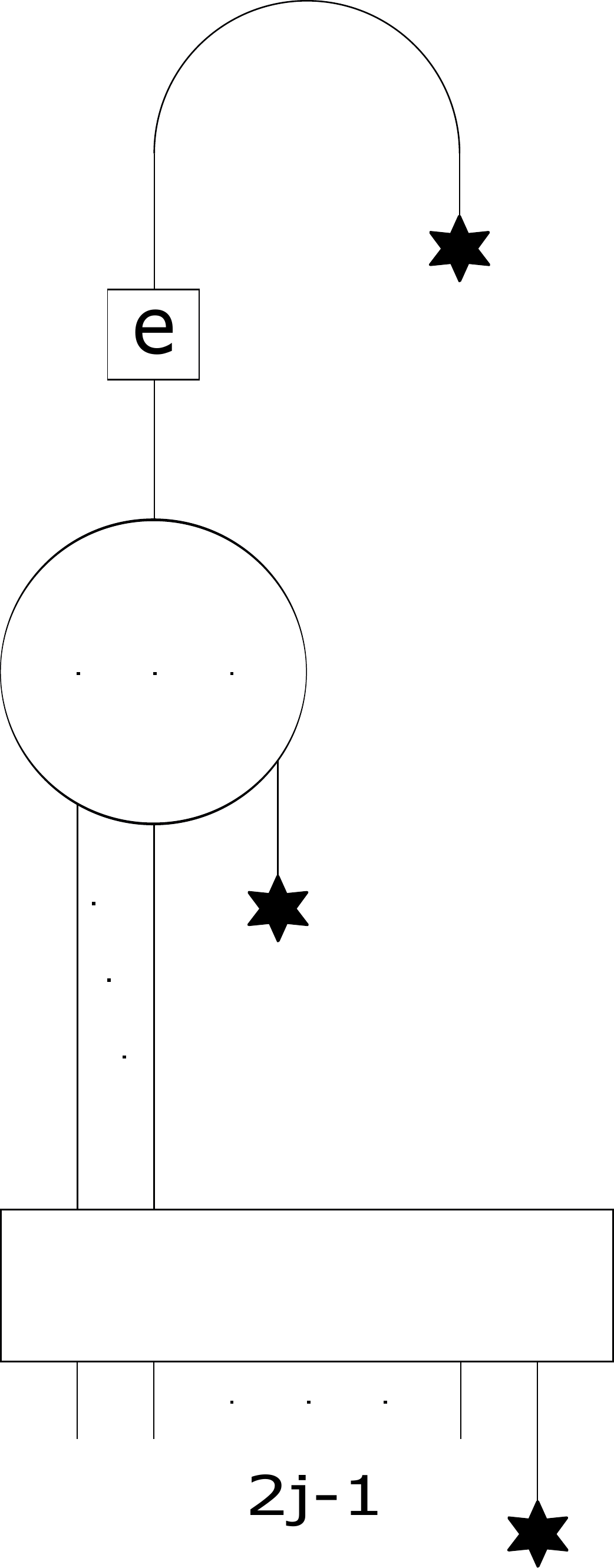}_{n+1} &\xrightarrow{\SprojF} - \left( j+ 1 \right) \psi \cdot \includegraphics[valign=c, scale=0.12]{Spinj_downup_general.pdf}_{n+1}
    \label{eq:Jn+1_projectdown}
\end{align}
Here, $\psi$ stands for an arbitrary spin network state admitting gravitational degrees of freedom only.
\end{prop}
The $+1$ in the negative eigenvalue seems odd at first, but is consistent with the fact that the mixed term of the total spin operator for the singlet state (case $j = \frac{1}{2}$) needs a contribution of $- \frac{3}{2}$ in order to yield the total spin 0. As a result, we found a more general subclass of eigenstates of a projection operator compared to \autoref{sec:surfnorm}. 

While the squared total spin $\SE{E}{2}$ is well understood, we are having trouble finding a well behaving counterpart for the spin projection $\SE{z}{}$. Still, first attempts have brought reminiscent structures in special cases.

Anyway, we draw an important conclusion from the above analysis, namely that fermionic systems influence the spin quantum numbers of the neighbouring holonomies in a way that should be principally measurable. In the following section, we want to sketch a measurement which makes this statement more specific.

\section{The Influence of Fermion Spin on the Geometry of Spacetime}
So far, we learned that fermionic spin degrees of freedom influence the spin quantum numbers of the neighbouring holonomies, which are in turn a measure of surface area \cite{Rovelli:1994ge}. This interaction between geometry and matter should in principle be observable by measuring the area of two surfaces in the vicinity of the fermion. As a quantum of area is of the order of $l_p^2$ and hence small compared to measurable scales of area, we want to sketch a macroscopic version of this experiment. Therefore, consider a cube $\mathcal{B} \subset \Sigma$ with a number of $n$ fermions each being attached to a single vertex of the spin network. We will argue later, however, that the details of the arrangement of the fermions are not so important for the effect we are discussing.

If there were no fermions, the Gauss law would constrain the spins flowing out of $\mathcal{B}$ to be able to couple to zero. With fermions being placed inside $\mathcal{B}$, the outflowing spins have to be able to couple to the total spin of the $n$ fermion system, which we can now measure with $\SE{E}{2}$. The simplest example of this scenario is described by two outflowing spins shown in \autoref{fig:effect_simple}, which depicts the lateral cross section of the cube. 
\begin{figure}[t]
    \centering
    \begin{align}
        j \quad &\underset{B}{\includegraphics[valign=c, scale=0.06]{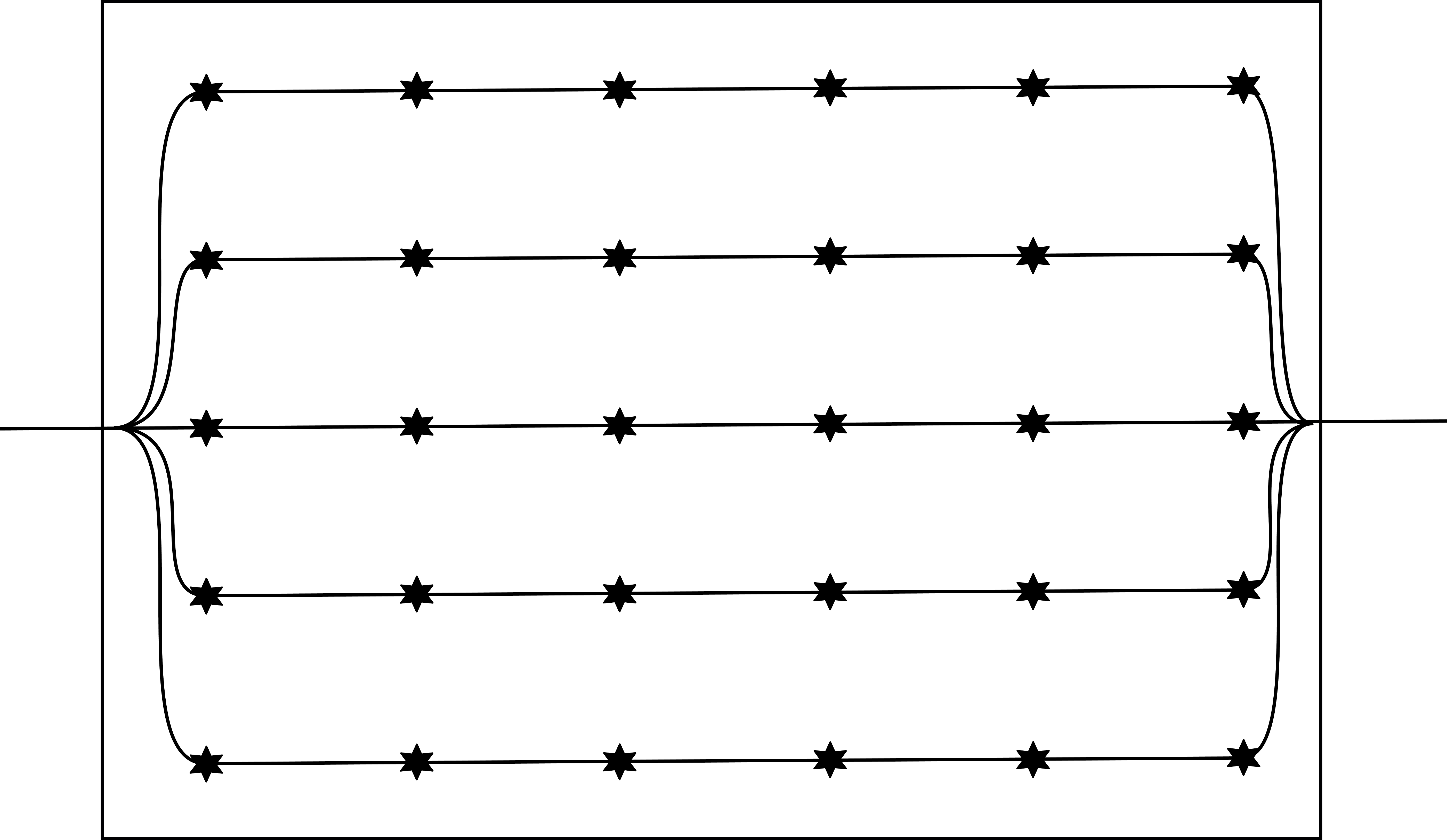}} \quad j + k \nonumber \\[3 ex]
        &\text{with } k \in \begin{cases}
            \big\{0, ..., \frac{n}{2}\big\} &\text{for } n \text{ even} \\
            \big\{ \frac{1}{2}, ..., \frac{n}{2} \big\} & \text{for } n \text{ odd}
        \end{cases}
        \label{eq:effect_simple}
\end{align}
    \caption{Sketch of the lateral cross section of a three-dimensional cube $\mathcal{B}$ containing $n$ fermions. Two egdes are intersecting the surface area of $\mathcal{B}$. Dependent on the coupling of the fermion spins, the spin quantum numbers of the holonomies intersecting the surface of $\mathcal{B}$ may differ by $k$, which is bounded by $0$ and $\frac{n}{2}$. This is reminiscent of the total spin of the Ising model in flat spacetime.}
    \label{fig:effect_simple}
\end{figure}

In general, the quantum state corresponding to \autoref{fig:effect_simple} can be described by a linear combination of binor calculus states like (\ref{Spinj_eigenstate}). The value $k$ can be tuned by aligning neighbouring spins anti-paralelly for $k$ minimal or paralelly for $k$ maximal, respectively. This scenario is a simple approach to mimic antiferromagnetism (or ferromagnetism respectively) in loop quantum gravity\footnote{In principle, $k$ can also take negative values. However, as the problem is symmetric, negative $k$ values can be equivalently described by flipping the outflowing spins.}. Using the well understood action of the area operator on holonomies \cite{Rovelli:1994ge} 
\begin{align}
    \hat{A}(S) \, \jrep{j}{h_e} &= \sqrt{\hat{E}_i \hat{E}^i} \, \jrep{j}{h_e} = \nonumber \\
    &= \left( 8 \pi l_p^2 \gamma \right) \sqrt{j(j+1)} \, \jrep{j}{h_e},
\end{align}
we notice an increase of area between the areas of the two surfaces on the left and on the right of \autoref{fig:effect_simple}. The difference reads
\begin{align}
    \Delta A &= \left( 8 \pi l_p^2 \gamma \right) \left( \sqrt{(j+k)(j+k+1)} - \sqrt{j(j+1)} \right) \nonumber \\
    &\approx \left( 8 \pi l_p^2 \gamma \right) k \quad \text{for  } j \gg 1.
    \label{eq:area_diff}
\end{align}
Hence, the maximal difference is proportional to the number of fermions $n$. Considering a cubic centimetre with an electron density of around $10^{23} \frac{1}{\cm^3}$ (average solid state), a rough dimensional analysis yields an approximate maximal area difference of
\begin{align}
    \Delta A \approx \gamma \cdot 10^{-42} \cm^2,
    \label{eq:area_difference}
\end{align}
being significantly larger than the Planck length squared $l_p^2$, but still small compared to area scales of well known physics. If we assume a homogeneous growth of area, we can trace it back to the error of length scale $\epsilon$
\begin{align}
    A_0 + \Delta A &= (l + \epsilon)^2 = l^2 + 2l\epsilon + \mathcal{O}(\epsilon^2) \\
    \implies \epsilon &\approx \frac{\Delta A}{2l}
\end{align}
where $l$ is the length of the square at total spin 0 and $A_0 = l^2$. Only for $l = 10^{-9} \cm$, $\epsilon$ enters the regime of Planck length and would be even smaller than $l_p$ for larger $l$. Although the large number of fermions first appeared to make the effect large, we end up with an effect in length of the order of the Planck length. This comes from the fact that it is apparently harder to measure a quantum of area $\sim l_p^2$ than a quantum of length $\sim l_p$.

The area change (\ref{eq:area_difference}) can be additionally increased by increasing the electron density or starting with a longer bar instead of a cube in the first place. Also a more efficient method to measure area (without tracing it back to length) might improve the problem of spatial resolution discussed above. In this way it might in principle be possible to experimentally measure an upper bound to the Barbero-Immirzi parameter $\gamma$. 

Indeed, this is consistent with the fact that the Barbero-Immirzi parameter in loop quantum gravity is no more arbitrary in the presence of fermions \cite{Perez:2005pm}. It might still be difficult to make a statement about the order of magnitude of $\gamma$. While the Immirzi parameter is estimated to be of order of magnitude 1 in order to reproduce the Bekenstein-Hawking formula in the theory of black hole entropy \cite{Domagala:2004jt, Meissner:2004ju}, it is treated as a large quantity $\gamma \gg 1$ in order to justify the perturbative calculation of time evolution \cite{Assanioussi:2017tql}.

\begin{figure}[t]
    \centering
    \includegraphics[scale=0.07]{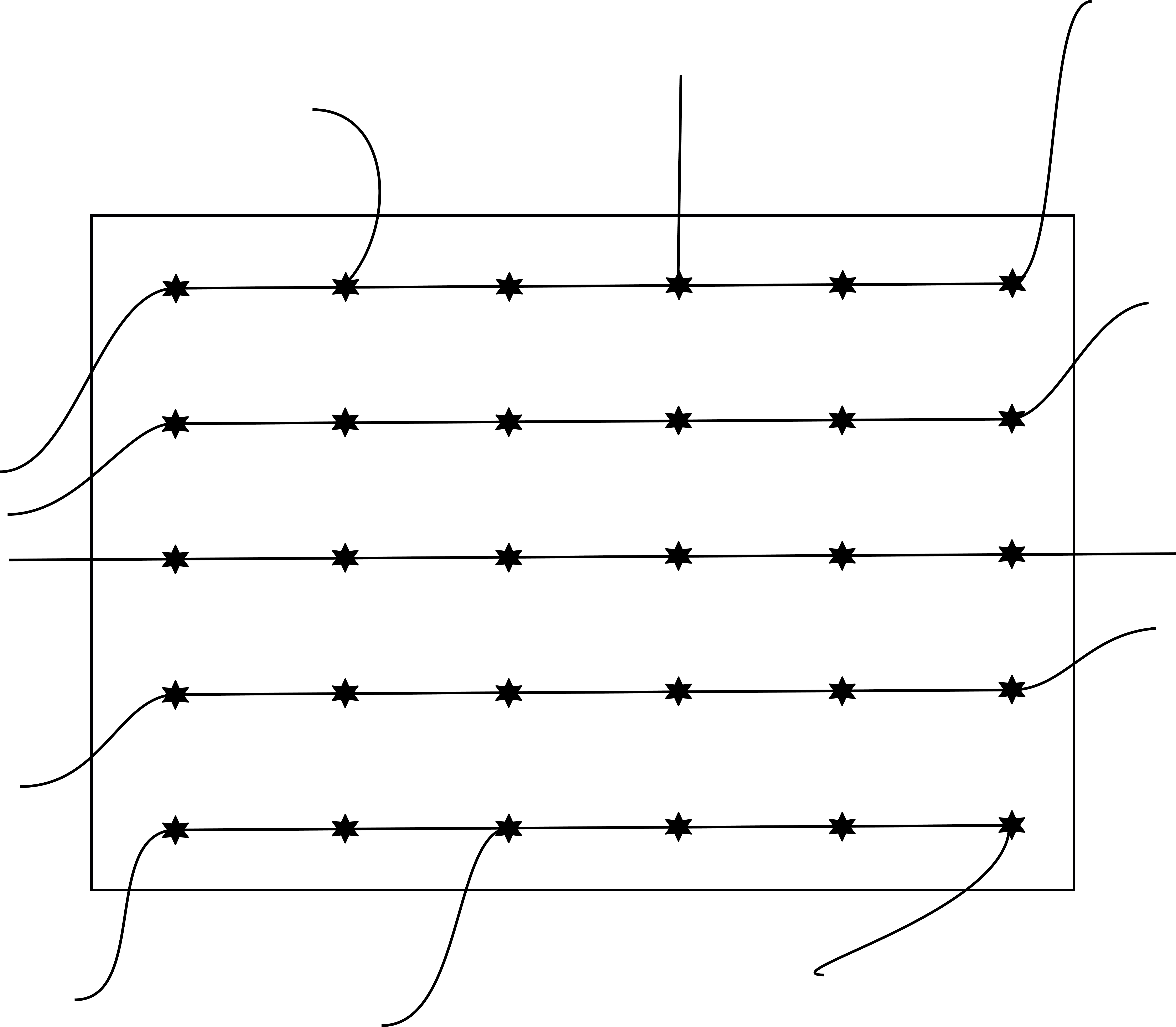}
    \caption{Sketch of the lateral cross section of a compact subset $\mathcal{B} \subset \Sigma$ of the spatial hypersurface, which is punctured by many holonomies. The total spin flowing out of (or into) $\mathcal{B}$ is a measure of the surface area increase compared to the surface area of $\mathcal{B}$ without matter.}
    \label{fig:Effect_general}
\end{figure}

Note that the above analysis only represents a model which would be more accurate if we would allow more than only two points of outflow. This includes states of the form (\ref{Spinj+k_ugly}), which are also eigenstates of the total spin squared. Here, however, the area increase is not directed like in (\ref{eq:effect_simple}) but might flow out of the cube in every direction. Still the minimal and maximal increase are again limited by $0$ and $\frac{n}{2}$. 

More generally, in order to predict the effect it is sufficient to require $\mathcal{B}$ to be a connected and compact subset of $\Sigma$ not necessarily diffeomorphic to a cube. The more general statement when embedding a spin network state with $n$ fermions into $\mathcal{B}$ is that the surface area of $\mathcal{B}$ will sense the same area increase (\ref{eq:area_diff}) when coupling the fermion spins completely parallel compared to the case of total spin 0. The most general setup is sketched in \autoref{fig:Effect_general}.

\section{Summary and Outlook}
We have introduced a total spin operator $S_E$ for a set of points connected by a set of paths $E$ to a reference point $t(E)$. The parallel transport of each fermionic spin operator to $t(E)$ ensures that we can add all the different single particle spin operators to a gauge covariant total spin operator. 

Building on $S_E$ we further introduced a number of observables describing the squared total spin and various projections of $S_E$. We showed that, in general, spin network states are not eigenstates of these operators. Instead, it turned out to be convenient to work with the binor formalism. Within this framework, we were able to describe the eigenstates of the squared total spin operator in some detail and hence define a notion of fermionic spin coupling in loop quantum gravity. $E_S$ is algebraically a spin, hence the spectrum of this operator is the standard one. 

Although we were not able to define a spin projection that commutes with the squared total spin, we introduced two possible observables which mimic the action of $\hat{S}_z$ from quantum mechanics, using on the one hand a fixed surface $\mathcal{S}$ and on the other hand the spin of one specific fermion as a reference. A more accurate imitation of a Stern-Gerlach type observable would be the coupling of the spin operator with an electromagnetic field. The formulation of loop quantum gravity with charged fermions is in principle available \cite{Thiemann:1997rq}, a detailed treatment is still work in progress \cite{mansuroglu2020kinematics}, however.

At last, we used the above results to sketch an experimental setup describing an observable effect as a result of the entanglement between fermion spin and geometry. In the example considered, a closed surface $\mathcal{B}$ containing $n$ fermions would experience a change in area linear in $n$ when the spins of the fermions inside $\mathcal{B}$ were aligned. We note again, that neither area nor spin are Dirac observables, and thus the picture could change in a more complete approach.  

Indeed, it would be very interesting to take diffeomorphism and Hamiltonian constraint into account to understand the dynamics of the kinematical states discussed in this paper. To do this, a reduced phase space quantisation could be carried out \cite{Giesel:2007wn,Giesel:2012rb}.

Also, to understand the connection between area and total spin better, the eigenstates of $\SE{E}{2}$ and the action of the area operator on these are to be understood more precisely. Since the two observables should be measured simultaneously, their commutator has to be investigated.

We finally note that we chose the Ashtekar-Barbero connection to parallely transport spins. Other choices of parallel transport are clearly possible in principle. How they could be realised within loop quantum gravity, and how this would change the picture could also be investigated in the future. 

\begin{acknowledgments}
We thank Kristina Giesel and Thomas Thiemann for interesting discussions during the completion of this work. Moreover, we thank Thomas Thiemann for very detailed and valuable feedback.   

R.M. thanks the Elite Graduate Programme of the Friedrich-Alexander-Universität Erlangen-Nürnberg for its support and for encouraging discussions among its members and faculty. H.S. would like to acknowledge the contribution of the COST Action CA18108.
\end{acknowledgments}

\appendix
\section{SPIN COUPLING AND THE BINOR FORMALISM}
\label{Penrose_Calculus_appendix}
It turns out that within the spectral analysis of the spin operator (\ref{def_Jptot2}), it is convenient to leave the spin network basis of loop quantum gravity and describe the states of $\mathcal{H}$ using the binor calculus 
\cite{Penrose71angularmomentum:, Penrose:1987uia, DePietri:1996pj, Kauffman02}. In order to fix the notation, we recall the main definitions of the binor-formalism:
\begin{align}
    &\includegraphics[valign=c, scale=0.5]{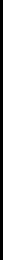}^A_B = \delta^A_B \hspace{5 ex} \leftidx{_A}{\includegraphics[valign=c, scale=0.25]{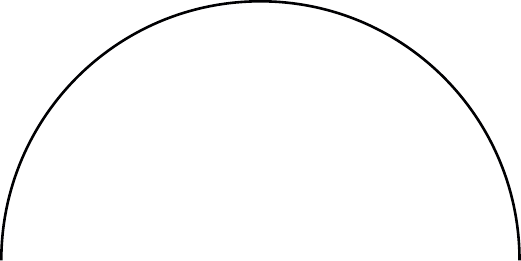}}{_B} = i \epsilon_{AB} \hspace{5 ex} \leftidx{^A}{\includegraphics[valign=c, scale=0.25]{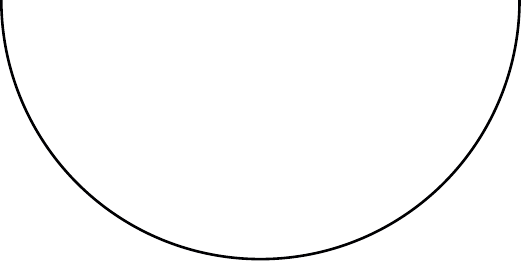}}{^B} = i \epsilon^{AB} \\[3 ex]
    &\leftidx{^A_C}{\includegraphics[valign=c, scale=0.5]{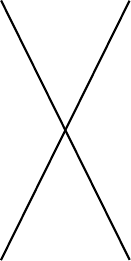}}{^B_D} = - \delta^A_D \delta^B_C \hspace{5 ex} \tensor{\includegraphics[valign=c, scale=0.5]{Edge_withSU2_rep.pdf}}{^A_B} = \tensor{\pi_{\frac{1}{2}}(h_e)}{^A_B} \hspace{5 ex} \includegraphics[valign=c, scale=0.7]{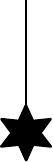}^A_i = \theta^A_i
\end{align}
For the sake of clarity, we deviate from the standard notation by denoting a spin $\frac{1}{2}$ representation of the holonomy $h_e$ as a straight line with label "e". Moreover, we draw a spinor as a star. Higher spin $j$ representations can be constructed by the symmetrization of $2j$ many spin $\frac{1}{2}$ holonomies. We denote the symmetrization by a box with as many ingoing edges as outgoing ones,
\begin{align}
    \includegraphics[valign=c, scale=0.2]{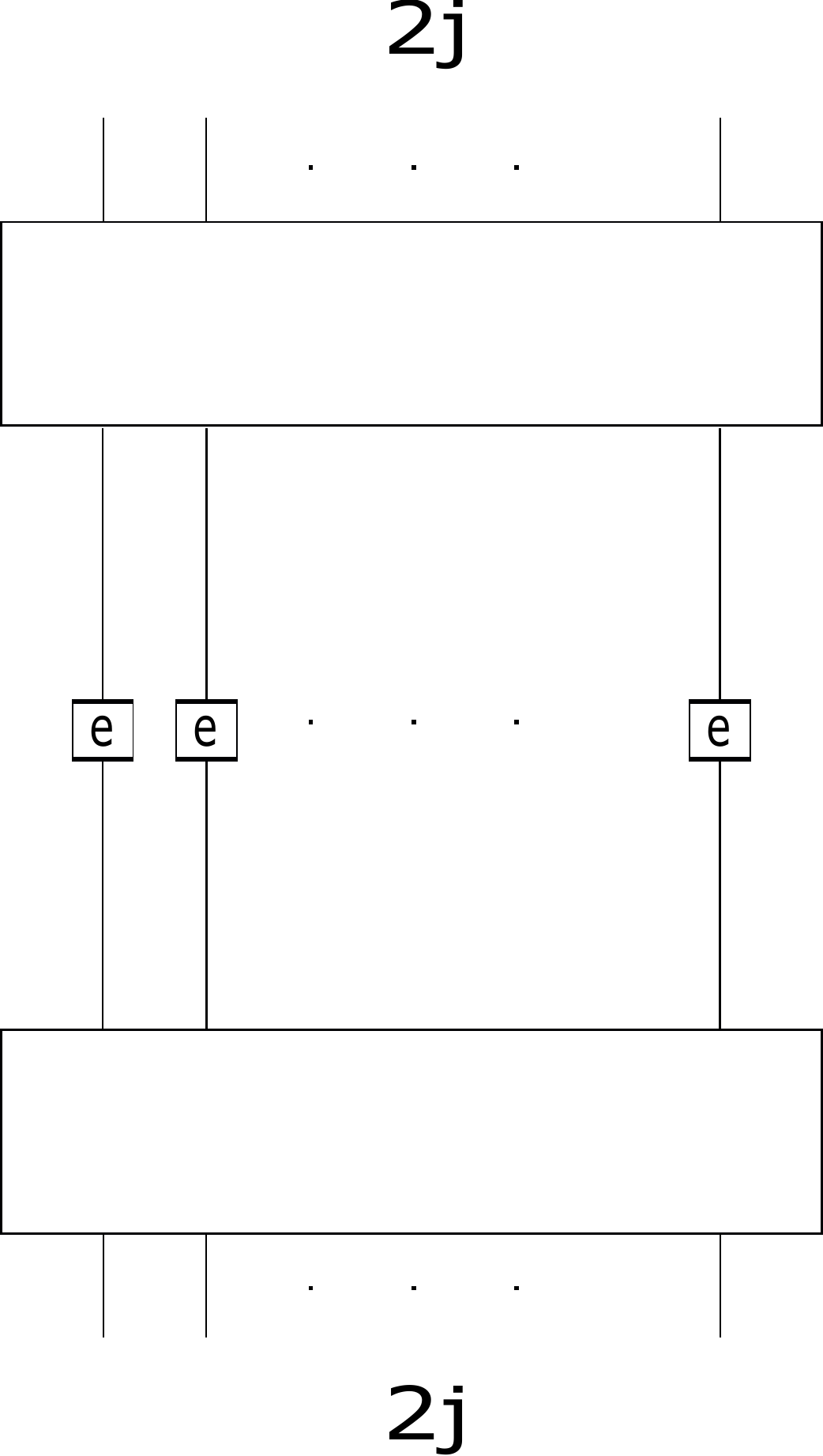} = \pi_{j}(h_e).
\end{align}
The last element from our toolbox which enables us to mimick spin network states in the binor formalism is the notation of the intertwining operators,
\begin{align}
    \includegraphics[valign=c, scale=0.15]{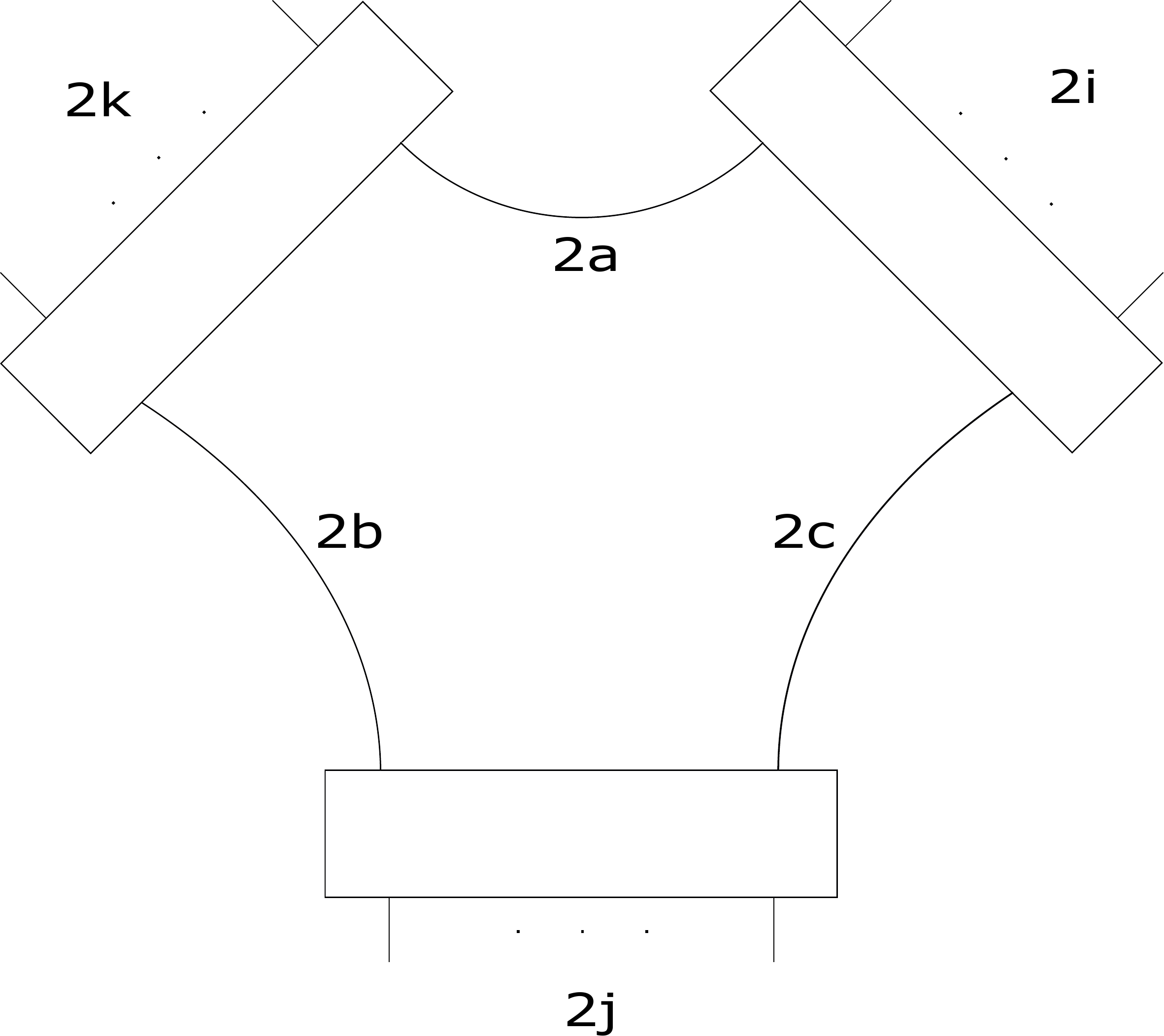} = \iota : i \otimes j \otimes k \rightarrow 0,
    \label{PC_intertwiner}
\end{align}
where $a, b, c \in \frac{\mathds{N}_0}{2}$ are determined by
\begin{align}
    a = \frac{1}{2} (i + k - j) \\
    b = \frac{1}{2} (j + k - i) \\
    c = \frac{1}{2} (i + j - k).
\end{align}
As higher valent intertwiners can always be expanded into a basis of 3-valent intertwiners, (\ref{PC_intertwiner}) completes the toolbox of the binor formalism. At last, we want to recall the binor identity,
\begin{align}
    \includegraphics[valign=c, scale=0.5]{Binor_1.pdf} \, + \, \includegraphics[valign=c, scale=0.5]{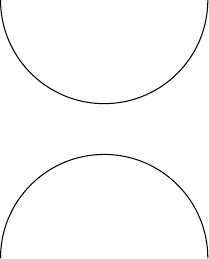} \, = - \, \includegraphics[valign=c, scale=0.5]{Binor_3.pdf} \iff \delta^A_B \delta^C_D - \epsilon^{AC} \epsilon_{BD} = \delta^A_D \delta^C_B,
    \label{eq:binor_identity}
\end{align}
which is used to dissolve crossings within the binor calculus. We can use the binor identity, to show the binor representation of (\ref{PC_2J1J2}). We start with the identity,
\begin{align}
    \frac{1}{2} \tensor{\sigma}{_i^A_B} \tensor{\sigma}{^i^C_D} = \leftidx{_B^A}{\includegraphics[valign=c, scale=0.1]{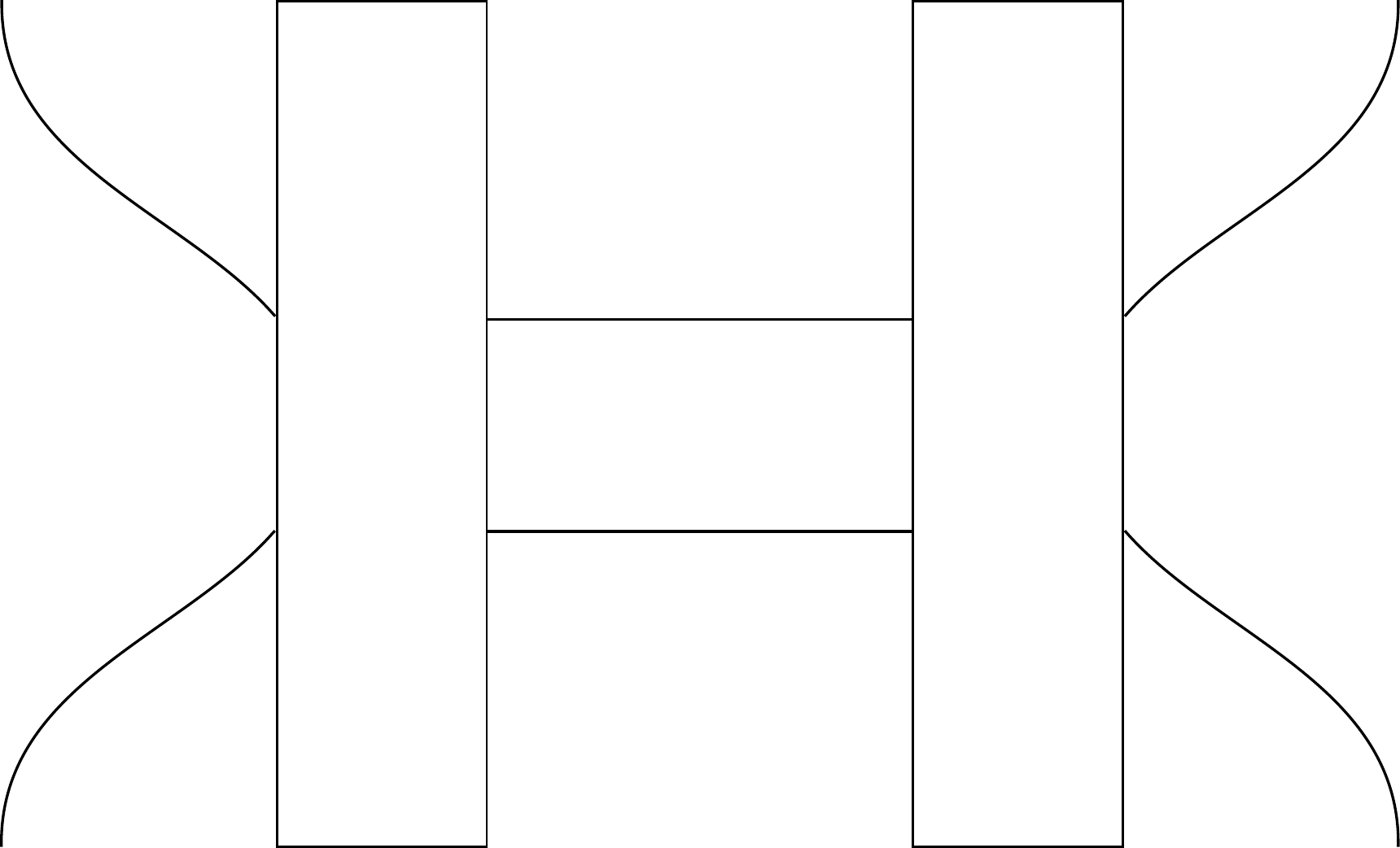}}{^C_D} = \frac{1}{2} \leftidx{^A_B}{\includegraphics[valign=c, scale=0.5]{Edge.pdf}}{} \quad  \leftidx{}{\includegraphics[valign=c, scale=0.5]{Edge.pdf}}{^C_D} + \leftidx{^A_B}{\includegraphics[valign=c, scale=0.5]{Binor_2.pdf}}{^C_D}
\end{align}
On the other hand, using the intertwining property of the Pauli matrices,
\begin{align}
    \tensor{\pi_1(h_e)}{^i_j} = \frac{1}{2} \tensor{\sigma}{^i^A_B} \tensor{\pi}{_{\frac{1}{2}}^B_C} \tensor{\sigma}{_j^C_D} \tensor{\left(\pi^{-1}_{\frac{1}{2}}\right)}{^D_A}
\end{align}
and substituting $\pi_1(h_e)$ in the definition (\ref{PC_2J1J2}), we finally deduce 
\begin{widetext}
\begin{align}
    &2 \tensor{\left(\hat{J}_{1j}\right)}{^A_B} \tensor{\pi_1\left(h_e \right)}{^j_k} \tensor{\left(\hat{J}_2^k \right)}{^C_D} = \frac{1}{4} \tensor{\sigma}{_j^A_B} \tensor{\sigma}{^j^E_F} \tensor{\pi}{_{\frac{1}{2}}^F_G} \tensor{\sigma}{_k^G_H} \tensor{\left(\pi^{-1}_{\frac{1}{2}}\right)}{^H_E} \tensor{\sigma}{^k^C_D} \nonumber \\[3 ex]
    &= - \quad \includegraphics[valign=c, scale=0.1]{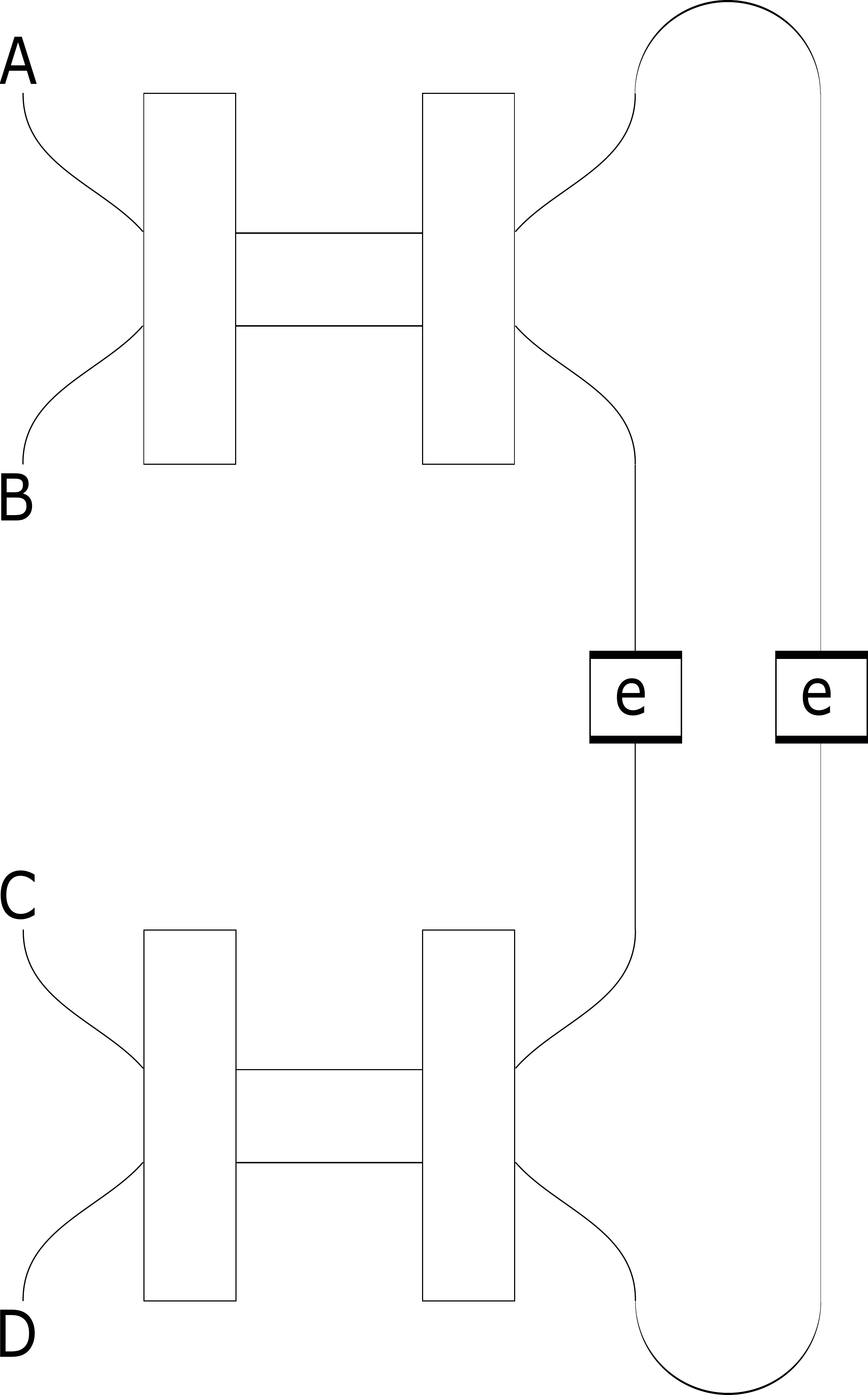} = \frac{1}{2} \, \leftidx{^A_B}{\includegraphics[valign=c, scale=0.5]{Binor_1.pdf}}{^C_D} + \leftidx{^A_B}{\includegraphics[valign=c, scale=0.5]{Binor_2_SU2_rep.pdf}}{^C_D} \nonumber \\[3 ex] 
    &= \leftidx{^A_B}{\includegraphics[valign=c, scale=0.15]{Pauli_squared_SU2_rep.pdf}}{^C_D} \label{PC_2J1J2_proof}
\end{align}
\end{widetext}

which shows (\ref{PC_2J1J2}). The operator (\ref{PC_2J1J2_proof}) acts on the fermions 1 and 2 by inserting a Pauli matrix at the point where the fermion is located. With this operator, we are able to calculate the action of $\left(S_E\right)^2$ on an arbitrary state in $\mathcal{H}$.

\section{PROOF OF PROPOSITION \ref{prop:ugly_eigenstates}}
\label{sec:proof_ugly_states}
We proceed with a proof of induction. 
\subsection*{Induction start}
We will consider example \ref{ex:Spin1_ugly} for the start. Here, we have $j = k = \frac{1}{2}$ and two fermions of different types. Apart from $\lambda_0$, which can always be set to 1 without loss of generality, there is only one non-trivial coefficient $\lambda_1 = \frac{1}{2}$, or equivalently $\lambda_1^{-1} = 2$, which is in agreement with (\ref{eq:coefficients}). 

\subsection*{Induction step}
We assume that the proposition holds for a fixed number of fermions $n$. If we now consider a state with $n+1$ fermions, there are four possible cases, we have to distinguish. On the one hand, we can lower or raise the total spin of the state and on the other hand we can either raise (or lower) $j$ or $k$. However, the problem is highly symmetric such that we can focus on the case where $j$ is raised by the $n+1^\text{st}$ fermion and refer to the other cases to be treated in analogue manner.

Moreover, we will reduce the problem further by dissolving all symmetrisations (except at the points of outflow) and operating on a product of fermions and holonomies only. Hence, we are left with the following building blocks
\begin{widetext}
\begin{align}
    \includegraphics[valign=c, scale=0.5]{Fermion.pdf} \, \includegraphics[valign=c, scale=0.5]{Fermion.pdf} \, ... \, \includegraphics[valign=c, scale=0.5]{Fermion.pdf}, \qquad \includegraphics[valign=c, scale=0.2]{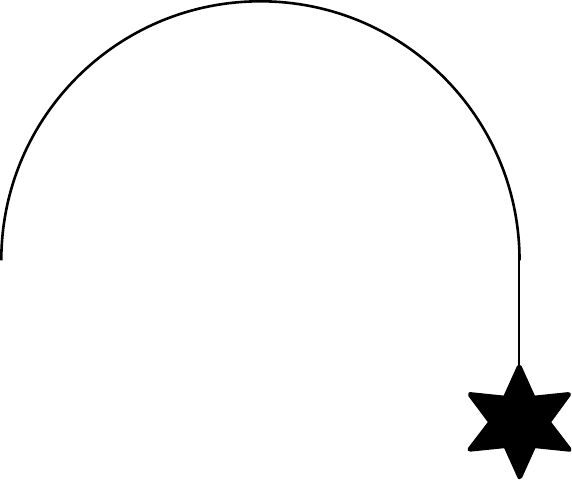} \, \includegraphics[valign=c, scale=0.2]{Fermion_Epsilon.pdf} \, ... \, \includegraphics[valign=c, scale=0.2]{Fermion_Epsilon.pdf}, \qquad \includegraphics[valign=c, scale=0.2]{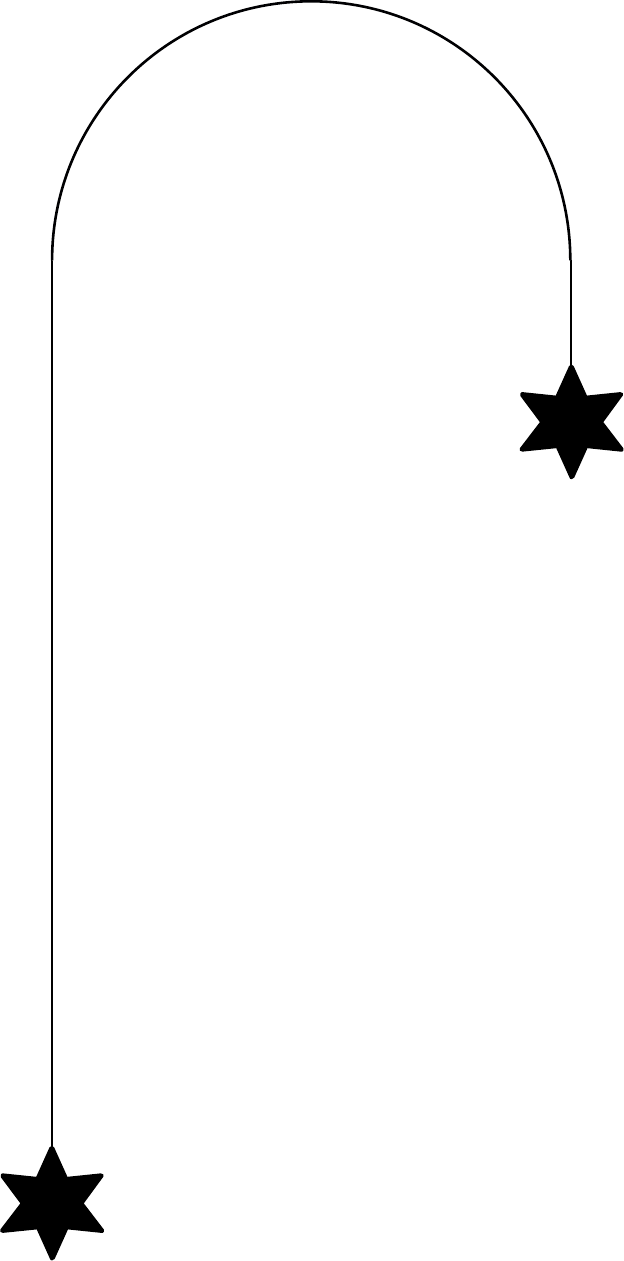} \, \includegraphics[valign=c, scale=0.2]{Spin0_wo_e.pdf} \, ... \, \includegraphics[valign=c, scale=0.2]{Spin0_wo_e.pdf},
\end{align}
\end{widetext}

out of which the remaining terms are constructed. We omitted the labels denoting the edges of the connecting holonomies for the sake of clarity. In our case, the $n+1^\text{st}$ fermion is of the form
\begin{align}
    \includegraphics[valign=c, scale=0.2]{Fermion_Epsilon.pdf}_{n+1.}
    \label{eq:increase_j}
\end{align}
If we denote the graph connecting fermion 1 to fermion $n$ to a joint point by $E'$ and $E = E' \cup \{e_{n+1}\}$, then we can split off the action of the squared total spin in the following way \\
\begin{align}
    \SE{E}{2} = &\SE{E'}{2} + \Sx{e_{n+1}(0)}{2} \nonumber \\
    &+ 2 \left( \sum_{l = 1}^n \Sx{e_l(0)}{i} \delta_{ij} \tensor{\jrep{1}{h_{e_{n+1}^{-1} \circ e_l}}}{^j_k} \right) \Sx{e_{n+1}(0)}{k}.
    \label{eq:tot_squared_spin_split}
\end{align}
If we act on the $n+1$ fermion state with the first operator, it will give the eigenvalue $j_m(j_m+1)$ by the induction hypothesis, where $j_m = j + k - m$ denotes the sum of outflowing spin, which is dependent on the number $m$ of reductions made. The second operator is also well known and gives the eigenvalue $\frac{1}{2} \left( \frac{1}{2} + 1 \right) = \frac{3}{4}$. We are left with the mixed terms, which connect the $n+1^\text{st}$ fermion with any of the other fermions each in a separate addend. The binor representation of this operator is shown in (\ref{PC_2J1J2}). One can see that we get two terms for each of its action if we dissolve the symmetrisation in (\ref{PC_2J1J2}) together with using the binor identity (\ref{eq:binor_identity}). The first term in (\ref{PC_2J1J2}) acts as an identity operator with a factor $\frac{1}{2}$. From this contribution, we get, additionally to $j_m(j_m+1)$ and $\frac{3}{4}$, also a factor $\frac{n}{2}$. We will now discuss the only non-trivial action of (\ref{eq:tot_squared_spin_split}) namely the second term in (\ref{PC_2J1J2}). This term connects the $n+1^\text{st}$ fermion with each of the building blocks as well as the endpoints which are left over by a spin $\frac{1}{2}$ holonomy.

Let us start with the building block (\ref{eq:increase_j}) which increases $j$. As the $n+1^\text{st}$ fermion is of the same type, the resulting state vanishes as the symmetrisation and an antisymmetrisation meet
\begin{align}
   \includegraphics[valign=c, scale=0.25]{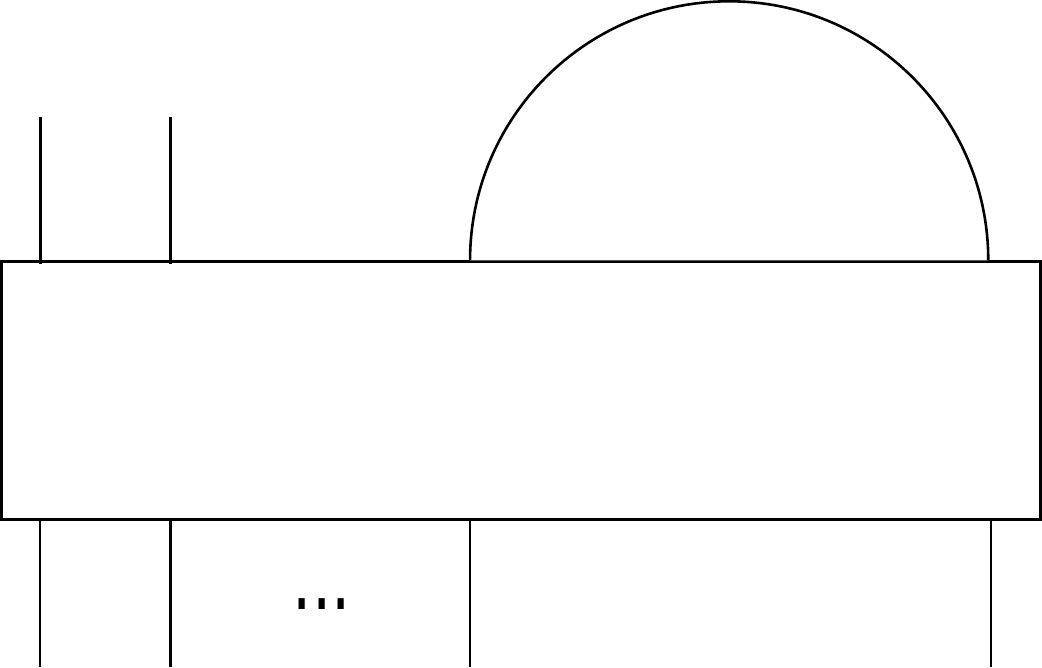} = 0.
\end{align}
This result is independent of which of the terms we regard. As a next step, let us have a look at the building blocks which are factors of singlet states,
\begin{align}
    \adjustbox{valign = c}{
        \begin{overpic}[scale = 0.2]{Spin0_wo_e.pdf}
            \put(-10, 0){\scriptsize 1}
            \put(32, 63){\scriptsize 2}
        \end{overpic}
    }_.
    \label{eq:singlet_building_block}
\end{align}
The operator connects the $n+1^\text{st}$ fermion with each of the two fermions 1 and 2 by a holonomy. This yields the state
\begin{widetext}
\begin{align}
    \includegraphics[valign=t, scale=0.2]{Fermion_Epsilon.pdf}_{n+1} \, \adjustbox{valign = c}{
        \begin{overpic}[scale = 0.2]{Spin0_wo_e.pdf}
            \put(10, 0){\scriptsize 1}
            \put(32, 63){\scriptsize 2}
        \end{overpic}
    } &\mapsto
    \includegraphics[valign=t, scale=0.2]{Fermion_Epsilon.pdf}_{2} \, \adjustbox{valign = c}{
        \begin{overpic}[scale = 0.2]{Spin0_wo_e.pdf}
            \put(10, 0){\scriptsize $n+1$}
            \put(32, 63){\scriptsize 1}
        \end{overpic}
    } + \includegraphics[valign=t, scale=0.2]{Fermion_Epsilon.pdf}_{1} \, \adjustbox{valign = c}{
        \begin{overpic}[scale = 0.2]{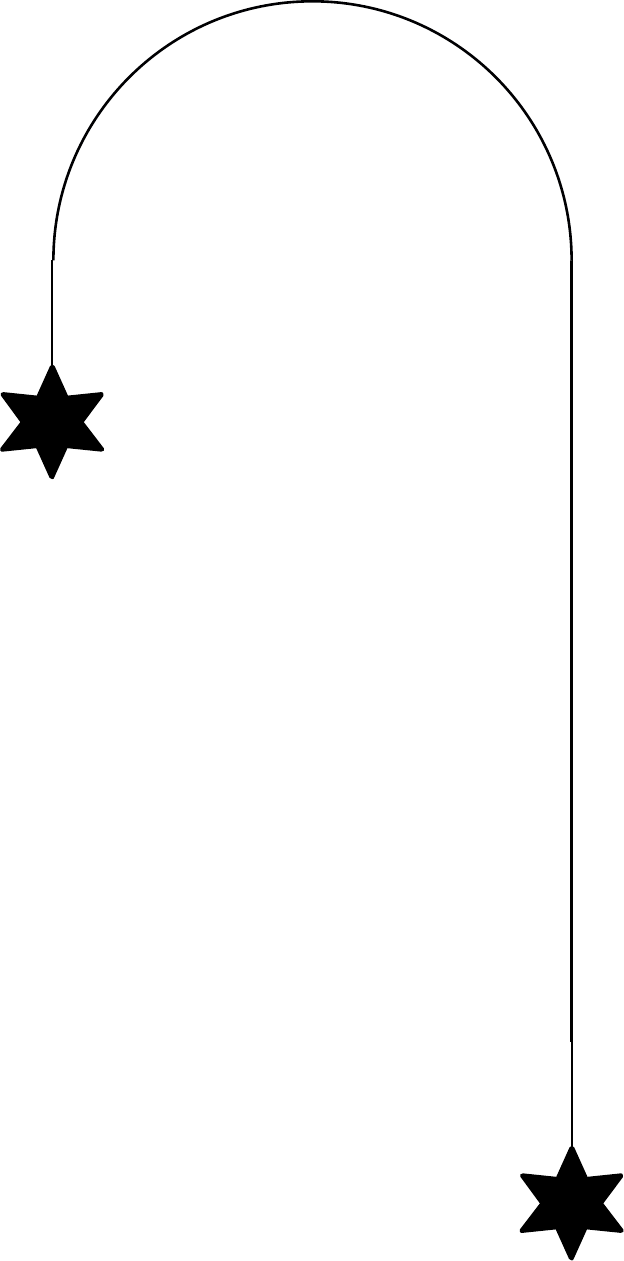}
            \put(50, 0){\scriptsize $n+1$}
            \put(-10, 63){\scriptsize 2}
        \end{overpic}
    } \nonumber \\[1 ex]
    &= - \quad \adjustbox{valign = c}{
        \begin{overpic}[scale = 0.2]{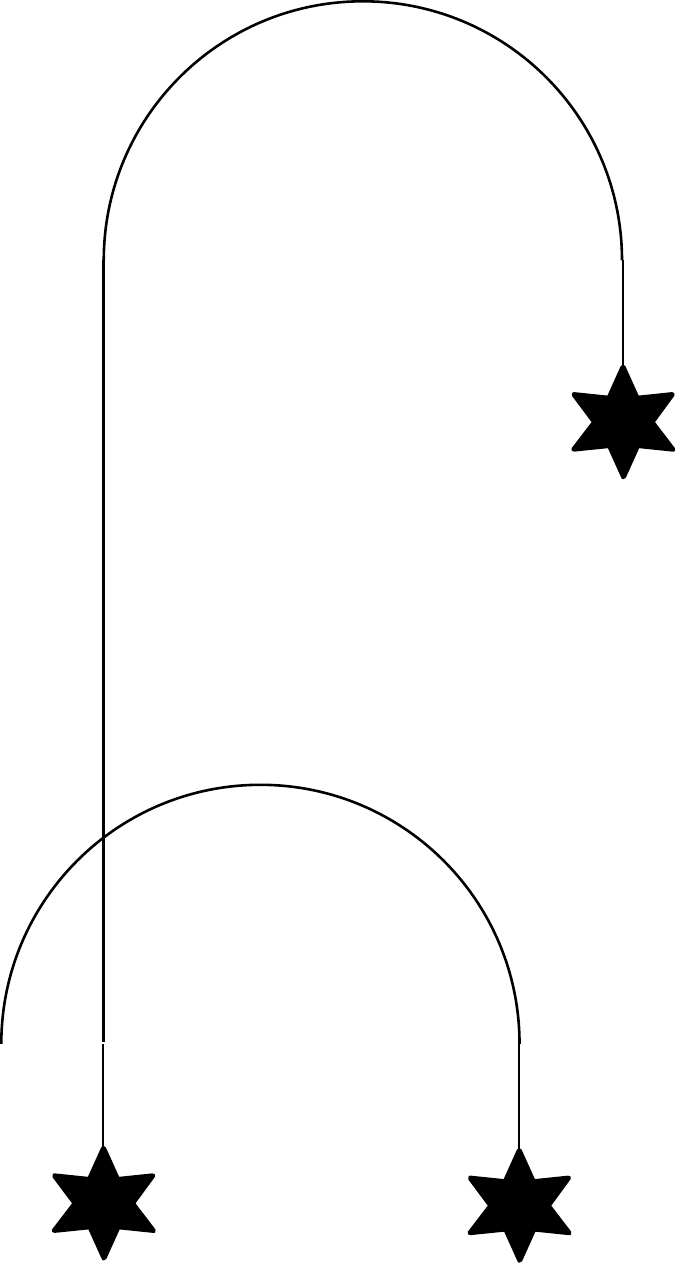}
            \put(-25, 0){\scriptsize $n+1$}
            \put(32, 63){\scriptsize 1}
            \put(50, 0){\scriptsize 2}
        \end{overpic}
    } - \includegraphics[valign=t, scale=0.2]{Fermion_Epsilon.pdf}_{1} \, \adjustbox{valign = c}{
        \begin{overpic}[scale = 0.2]{Spin0_wo_e.pdf}
            \put(10, 0){\scriptsize $n+1$}
            \put(32, 63){\scriptsize 2}
        \end{overpic}
    } \nonumber \\[1 ex] 
    &= \includegraphics[valign=t, scale=0.2]{Fermion_Epsilon.pdf}_{n+1} \, \adjustbox{valign = c}{
        \begin{overpic}[scale = 0.2]{Spin0_wo_e.pdf}
            \put(10, 0){\scriptsize 2}
            \put(32, 63){\scriptsize 1}
        \end{overpic}
    } = - \includegraphics[valign=t, scale=0.2]{Fermion_Epsilon.pdf}_{n+1} \, \adjustbox{valign = c}{
        \begin{overpic}[scale = 0.2]{Spin0_wo_e.pdf}
            \put(10, 0){\scriptsize 1}
            \put(32, 63){\scriptsize 2}
        \end{overpic}
    }_,
    \label{eq:proof_singlet_block}
\end{align}
\end{widetext}

where we used basic manipulations of $\epsilon$ and $\delta$ and the binor identity in the second step. The labels of the edges are again omitted, since they are not relevant for this calculation. In general, they have to be taken care of. With (\ref{eq:proof_singlet_block}) we have another contribution to the eigenvalue, which goes linear with the number of the building blocks of the form (\ref{eq:singlet_building_block}). Let us denote the number of fermions which decrease the number of holonomies from top to bottom by $n_-$. Analogously, the number of fermions increasing the number of holonomies, we denote by $n_+$ such that they sum up to the total number of fermions excluding the $n+1^\text{st}$ fermion,\footnote{Whether we include it or not is a matter of convention.}
\begin{align}
    n_- + n_+ = n.
\end{align}
On the other hand, it holds $n_+ - n_- = 2j - 2k$. The contribution of (\ref{eq:proof_singlet_block}) hence will be $- (n_- - 2k + m)$, which can be combined with $\frac{n}{2}$ to
\begin{align}
    \frac{n}{2} - n_- + 2k - m &= \frac{1}{2} (n_+ - n_-) + 2k - m \nonumber \\
    &= j + k - m = j_0 - m.
\end{align}
The last building block will change the state non-trivially. If we connect the $n+1^\text{st}$ fermion with a fermion of the type
\begin{align}
    \includegraphics[valign=c, scale=0.5]{Fermion.pdf}_,
\end{align}
\newpage
\noindent then we reduce the outflowing spin by 1. This reduction yields
\begin{align}
    \includegraphics[valign=t, scale=0.2]{Fermion_Epsilon.pdf}_{n+1} \includegraphics[valign=b, scale=0.5]{Fermion.pdf} \, \mapsto \, \includegraphics[valign=c, scale=0.7]{Edge.pdf} \leftidx{_{n+1}}{\includegraphics[valign=c, scale=0.2]{Spin0_wo_e.pdf}}{_,}
\end{align}
which is equal to one of the other states corresponding to one more reduction $m+1$. This contribution goes with a factor 1. Finally, we can collect all the pieces to find that the coefficients in front of a state with $m$ reductions reads
\begin{align}
    &\lambda_{m-1} + \lambda_m \left( (j_0 - m)(j_0 - m + 1) + \frac{3}{4} + j_0 - m \right) = \nonumber \\
    = &\lambda_{m-1} + \lambda_m \left( \left(j_0 + \frac{1}{2}\right) \left(j_0 + \frac{3}{2} \right) + m^2 - 2m - 2mj_0 \right),
\end{align}
which equals to $(j_0 + \frac{1}{2})(j_0 + \frac{3}{2})$ if and only if
\begin{align}
    \lambda_{m-1} = \lambda_m \cdot m \left( 2 \left( j_0 + \frac{1}{2} \right) - m + 1 \right).
\end{align}
This completes the proof. \pushQED{\qed} \popQED

\twocolumngrid
\bibliographystyle{abbrv}
\bibliography{literature}

\end{document}